\documentclass[12pt,reqno]{amsart}
\usepackage{fullpage}
\usepackage{times}
\usepackage{graphicx}

\usepackage{amsmath,amsthm,amsfonts,amscd,amssymb}
\usepackage{pst-node}
\usepackage{pst-plot}

\usepackage{a4}
\usepackage{amssymb}
\usepackage{array}
\usepackage{booktabs}
\usepackage{multirow}
\usepackage{hhline}

\newtheorem{thm}{Theorem}

\newtheorem{lem}[thm]{Lemma}
\newtheorem{prop}[thm]{Proposition}

\newtheorem{defn}[thm]{Definition}
\newtheorem{ex}[thm]{Example}

\newtheorem{rem}[thm]{Remark}

\numberwithin{thm}{section}
\numberwithin{equation}{section}
\newcommand{\Real}{\mathbb R}

\newcommand{\eps}{\varepsilon}

\newcommand{\la}{\langle}
\newcommand{\ra}{\rangle}

\newcommand{\B}{\mathcal{B}}

\newcommand{\Comp}{\mathbb{C}}

\newcommand{\D}{\mathcal{D}}

\newcommand{\F}{\mathcal{F}}
\newcommand{\G}{\mathcal{G}}

\newcommand{\Hi}{\mathcal{H}}
\newcommand{\I}{\mathcal{I}}

\newcommand{\n}{\mathbb{N}}

\newcommand{\tor}{\mathbb{T}}

\newcommand{\W}{\mathcal{W}}

\newcommand{\z}{\mathbb{Z}}

\newcommand{\Om}{\Omega}
\newcommand{\om}{\omega}

\newcommand{\ch}{\hat{c}}

\newcommand{\xv}{\mathbf{x}}

\newcommand{\yv}{\mathbf{y}}
\newcommand{\zv}{\mathbf{z}}

\newcommand{\weyl}{{W}}

\newcommand{\mc}[1]{\mathcal #1}

\newcommand{\tr}{{\rm Tr}\,}

\newcommand{\be}{\begin{equation}}
\newcommand{\ee}{\end{equation}}
\newcommand{\bs}{\begin{split}}
\newcommand{\es}{\end{split}}

\begin{document}
\title{Twisted Fourier analysis and pseudo-probability distributions}

\author{Sang Jun Park}
\author{Cedric Beny}
\author{Hun Hee Lee}

\address{Hun Hee Lee :}
\email{hunheelee@snu.ac.kr}

\keywords{twisted Fourier transform, wigner function, pseudo-probability distributions}
\thanks{2000 \it{Mathematics Subject Classification}.
\rm{Primary 81P45, 43A65}}

\begin{abstract}
We use a noncommutative generalization of Fourier analysis to define a broad class of pseudo-probability representations, which includes the known bosonic and discrete Wigner functions. We characterize the groups of quantum unitary operations which correspond to phase-space transformations, generalizing Gaussian and Clifford operations. As examples, we find Wigner representations for fermions and hard-core bosons.
\end{abstract}

\maketitle


\section{Introduction}

The Weyl-Wigner representation provides a formulation of quantum mechanics which mimics Hamiltonian classical mechanics. Quantum states are represented by real functions over a phase space (the Wigner function or pseudo-probability distribution) and can be treated essentially like probability distributions except for the fact that they can take negative values. 

This is useful in part because of the richness of the set of pure states with non-negative Wigner functions and dynamics preserving this property. 
By Hudson's theorem~\cite{hudson1974wigner}, they are the Gaussian states (squeezed coherent states) and the dynamics generated by quadratic Hamiltonians. 
These Gaussian states and dynamics are important not just because they provide a set of exact solutions to quantum mechanics, but also because these solutions are those which are most resistant to decoherence for bosonic systems~\cite{zurek1993coherent}, and hence closely describe many experimental setups (e.g. linear optics). For these reasons, gaussian states and unitaries are used to implement quantum information protocols~\cite{weedbrook2012gaussian}), and form the core of the perturbative formulation of bosonic quantum field theory (they are the free field formulations).

This formalism has a long history and has been adapted to other systems, or generalized in a variety of ways~\cite{Ali2000,ferrie2009framed,hayashi2017group,gibbons2004discrete,hinarejos2012wigner,klimov2017generalized,schwonnek2018wigner,raussendorf2019phase}. One of the most complete adaptation covers tensor products of systems of odd dimensions, with a theorem that parallel Hudson's, relating the corresponding non-negative discrete Wigner functions to stabilizer states, preserved by Clifford unitary maps~\cite{Gro}. As for Gaussian states and dynamics, this also leads to an important set of classically simulatable solutions~\cite{mari2012positive}. Moreover, they are also related to protection from decoherence, at least for qubits: stabilizer states where introduced specifically to build quantum error-correcting codes, and the Clifford operations can be implemented transversally, i.e., in a way which does not amplify errors. 

Here we develop a formalism based on twisted Fourier analysis~\cite{KleppnerLipsman72} which encompasses the above two examples of quasi-probability representations, as well as another formulation for angle-number systems~\cite{rigas2010non}.
In this framework, the ``phase space'' is a locally compact abelian group equipped with a certain 2-cocycle, which defines the twisting. We show how to obtain a Weyl-Wigner representation of the quantum states and observables.

The Wigner representation $\W_A$ for an operator $A$ which we obtain satisfies (under the assumption on the 2-cocycle being normalized) the following general properties (here summarized schematically ignoring for now the regularity assumptions):
\begin{enumerate}
\item $\W_A$ is real when $A$ is self-adjoint,
\item $A$ can be obtained from $\W_A$ by an inversion formula,
\item $\int \W_\rho(\xv) \W_A(\xv) d\xv = \tr(\rho A)$ for any observable $A$ and state $\rho$. 
\end{enumerate}

We also characterize the group of unitary transformations on the quantum states (which we refer to simply as {\em Clifford group}) corresponding to phase space transformations, which generalizes Gaussian and the standard qubit Clifford operations \cite[pp. 127--128]{Got}.

We show how to recover known results about bosonic modes and tensor products of odd-dimensional systems as examples of this definition. Moreover, we show how this formalism can be applied to fermions and hard-core bosons. For fermions, this yields a commutative discrete phase-space formalism, by contrast with the usual non-commutative Grassmann representation. For this example, we show that a generalization of the Clifford group for this system is isomorphic to that for qubits. 


\section{Background}
\label{section:background}

Quantum theory is plagued by the ``curse of dimensionality'', namely the exponential growth of the Hilbert space dimension as a function of the number of elementary systems. This makes simulations on a classical computer intractable. 

In quantum theory, a state is mathematically represented by a density matrix $\rho$, and an observable can be described by a self-adjoint operator $A$. Together, they can be combined to form a physical prediction in the form of an expectation value $\tr(\rho A)$.
By comparison, in a probabilistic classical theory, a states is represented by a positive function $\mu: \Omega \rightarrow \mathbb R^+$ over the phase space $\Omega$, and an observable by a real function $f$ on $\Omega$, which together can be combined to form an expectation value $\int_\Omega \mu(x) f(x) dx$. This can also be formulated in the language of operators, where the key simplification compared to quantum theory is the restriction of the states and observables to a commutative algebra.

A quasi-probability representation of quantum theory is one in which expectation values are obtained as in a classical theory, but where the functions representing the states can be negative. (See \cite{ferrie2011quasi} for a general definition along those lines). The canonical example is that of the Weyl-Wigner representation for a non-relativistic particle (an introduction can be found for instance in \cite{case2008wigner}), where the quantum state (density matrix) $\rho$ is represented by the Wigner function $\W_\rho(x,p)$, which is the (symplectic) Fourier transform of the characteristic function
\begin{equation}
\label{simple-wigner}
\chi_\rho(x,p) = \tr(W(x,p)^* \rho),
\end{equation}
where $W_\rho(x,p)$ is the unitary operator (called the {\em Weyl operator}) translating the particle by $x$, and boosting its momentum by $p$. More precisely, we have $W(x,p) = e^{i (p \hat x - x \hat p)}$, where $\hat x$ and $\hat p$ are the canonical position and momentum operators. Note that they form a projective representation of the group of phase-space translations $(\mathbb R^2, +)$.

Such a description allows one to identify a set of states and dynamics for which classical simulations methods can be used: those quantum states which have a non-negative quasi-probability distribution: $\W_\rho(x,p) \ge 0$ for all $x,p$, and the dynamics which preserve this property.

For instance, the pure states of a quantum particle with a non-negative Wigner functions are the Gaussian states (squeezed states), and this property is preserved by the Gaussian unitary maps (linear optics elements)~\cite{hudson1974wigner,soto1983wigner}. On many modes, this also covers quasi-free bosonic field theories. A similar formalism for tensor products of odd-dimensional systems also characterizes stabilizer states and Clifford operations~\cite{Gro}.

Let us give some context behind the definition of the Wigner function. 
The above mentioned Weyl operators $W(x,p)$ are bounded alternatives to the canonical position and momentum operators $\hat x$ and $\hat p$. They pertain to a way of formalizing Dirac's idea of {\em canonical quantization} which played an important role in the development of quantum mechanics and quantum field theory. In this approach, a quantum theory is defined from a classical one by looking for a Hilbert space equipped with an irreducible projective representation of the group of phase-space translations, i.e., unitaries $W(x,p)$ such that $W(x,p) W(x',p') = \sigma((x,p),(x',p')) W(x + x',p + p') $. The classical symplectic product between $(x,p)$ and $(x',p')$ appears in the phase $\sigma((x,p),(x',p')) = e^{-\frac i 2 (x p' - x' p)}$, enforcing the representation's non-commutativity. This also corresponds to a unitary representation of a central extension of the group $\mathbb R^2$---the (reduced) Heisenberg group---provided that $\sigma$ is a 2-cocycle, as it is.

Generally, we are interested in such representations because any group of symmetries of a physical system's {\em classical limit} ought to act at the quantum level as well. Although one can imagine more general situations, the simplest assumption is that the group should be represented by unitary operators up to a phase, hence that there exists a projective unitary representation at the quantum level. 

In the above example, one may in principle start from a larger group of phase-space symmetries, e.g. the whole group of canonical transformation, however the translations $\mathbb R^2$ together with the above 2-cocycle $\sigma$ turns out to already give us the right quantum systems for bosonic modes or non-relativistic particles. This is helped by the fact that it has a unique irreducble $\sigma$-representation, which implies that the only freedom left in defining the quantum system is the trivial one of adding extra unrelated variables.

This uniqueness result can be generalized to a large class of abelian groups and 2-cocycles~{\cite[Theorem 4]{DigernesVaradarajan04}}, a framework we adopt below. For instance, systems of finite-dimensional quantum systems can be obtain in a variety of ways in this manner, such as done in \cite{CCS}. This framework of quantization is interesting for our purpose because it constructs the quantum operators as functions over the group, but equipped with a non-commutative product (and completed in some topology). As we will see, with the right assumptions the (symplectic) Fourier transform of these functions have all the properties of Wigner representations of the corresponding operators.


The full theory is not essential to understand our results, but for completeness it is summarized in the Appendix~\ref{twistedfouriertransform}. 
In short, the non-commutative product introduced between ``characteristic functions'' is obtained by introducing the 2-cocyle $\sigma$ into the usual convolution product. Equipping these functions with the convolution would yield the commutative {\em group algebra}. But the {\em twisted} convolution yields, instead, the algebra of quantum observables. The mapping from functions to operators is a {\em twisted Fourier transform}~\cite{maillard1986twisted}. 

\color{black}



\section{Generalized Weyl-Wigner representation}\label{sec-gen-Weyl}

In this section we  develop a generalized Weyl-Wigner representation starting from abstracting {\em phase space} as a locally compact abelian group $G$, which we use additive notation $\xv+\yv \in G$ for the group operation for $\xv,\yv \in G$ with $0$ as the identity element. For instance, systems of $n$ bosonic modes, or $n$ non-relativistic distinguishable particles, will be recovered for the additive group $G = \mathbb R^{2n}$.

Our phase space $G$ is equipped with a {\em 2-cocycle} $\sigma: G\times G \to \tor$, which is a Borel function satisfying the conditions $\sigma(\xv,\yv)\sigma(\xv+\yv,\zv) = \sigma(\xv,\yv+\zv)\sigma(\yv,\zv)$, $\sigma(\xv,0) = \sigma(0,\yv) = 1$
for any $\xv,\yv,\zv\in G.$ For example, a usual choice of 2-cocycle on a system of $n$ bosonic modes
 with $G = \mathbb R^{2n}$ is given by
    \begin{equation}\label{eq-Euclid-2-cocycle}
    \sigma_{\rm boson}(\xv, \yv) =  \exp\left(-\frac{i}{2}\xv^T J \yv\right), \;\;\xv, \yv \in G, 
    \end{equation}
where $J={\footnotesize \begin{bmatrix}0 & I_n\\ -I_n & 0\end{bmatrix}}\in M_{2n}(\Real)$. Note that $(\xv, \yv) \mapsto  \xv^T J \yv$ is the canonical symplectic form on $\Real^{2n}$. Following this we would like to impose a kind of symplectic structure on $G$, namely we assume that the 2-cocycle $\sigma$ is an {\em Heisenberg multiplier}. This means that the map $\Phi: G \to \widehat{G}$ given by
    \begin{equation}\label{eq-symplectic-iso}\Phi(\xv)(\yv) = \sigma(\xv,\yv)\overline{\sigma(\yv,\xv)},\; \xv,\yv\in G
    \end{equation}
is a topological group isomorphism. Here, $\widehat G$ is the {\em dual group} of $G$ composed of the {\em characters} of $G$, i.e., continuous group homomorphisms $\chi: G \rightarrow \mathbb T$ into the circle group $\mathbb T = \{e^{i \theta} : \theta \in \mathbb R\}$. Note that $\Phi$ is in general different from the usual choice of isomorphism $x \in G\mapsto \gamma_x \in \widehat{G}$, which we call the {\em canonical identification}. From the fact that $\Phi(\xv)(\xv) = 1$ for any $\xv\in G$ the isomorphism $\Phi$ is called a {\em symplectic self-duality} of $G$ (\cite{PrasadShapiroVemuri10}).

With the above additional assumption on $\sigma$ we know (\cite[Theorem 4]{DigernesVaradarajan04}) that there is a unique irreducible unitary projective representation with respect to $\sigma$ (shortly, $\sigma$-representation) $W: G\to \mathcal{U}(\Hi_W)$ for some Hilbert space $\Hi_W$. Being $\sigma$-representation means that the map $\xv \mapsto W(\xv)h$ is Borel for any $h\in \Hi_W$ and we have
    \begin{equation}\label{eq-proj-rep}
    W(\xv)W(\yv) = \sigma(\xv,\yv)W(\xv+\yv),\; \xv,\yv \in G.
    \end{equation}
We call $W$ and $W(\xv)$, $\xv\in G$, the {\em Weyl representation} and the {\em Weyl operators} following the bosonic case. Now we can define characteristic functions of quantum states on a Hilbert space $\Hi
= \Hi_W$, which can be easily extended  to the case of trace class operators. Recall that the set of all quantum states on $\Hi$ (denoted by $\D = \D(\Hi)$) is a subset of $\mc S^1(\Hi)$, the trace class on $\Hi$ equipped with the trace norm $\|X\|_1 = {\rm Tr}(|X|) ={\rm Tr}((X^*X)^{\frac{1}{2}})$, $X\in \mc S^1(\Hi)$. Note that $\mc S^1(\Hi)$ is a subspace of $\mc S^2(\Hi)$, the Hilbert-Schmidt class on $\Hi$ equipped with the Hilbert-Schmidt norm $\|X\|_2 = ({\rm Tr}(X^*X))^{\frac{1}{2}}$, $X\in \mc S^2(\Hi)$.

	\begin{defn} \label{defn-chftn}
		Let $\rho \in \mc S^1(\Hi)$. We define its {\bf characteristic function} $\chi = \chi_\rho$ on $G$ by
			$$\chi_\rho(\xv) := {\rm Tr}(W(\xv)^*\rho),\;\; x\in G.$$
	\end{defn}

\begin{rem}
The terminology ``characteristic function'' can be justified from the fact that $\chi_\rho$ determines the original operator $\rho$ via the twisted Fourier transform $\F_\sigma$ on the group $G$. 
    $$\rho = \F_\sigma(\chi_\rho) := \int_G \chi_\rho(\xv) W(\xv) d\mu(\xv),\; \rho \in \mc S^1(\Hi).$$
Here, $\mu$ is the Haar measure on $G$ respecting the twisted Plancherel formula \eqref{eq-twisted-Plancherel} and we know that $\chi_\rho\in L^2(G)$. See Proposition \ref{prop-Fourier-char} for the details. Note that the integral $\int_G \chi_\rho(\xv) W(\xv) d\mu(\xv)$ can be understood as a bounded operator on $\Hi$ defined in the weak sense (Proposition \ref{prop-op-integral-bdd}).
\end{rem}	
	
Now we move to the definition of (abstract) Wigner functions of quantum states.
	
	\begin{defn} \label{defn-Wignerftn}
		Let $\rho \in \mc S^1(\Hi)$. We define its {\bf Wigner function} $\W = \W_\rho: G \to \Comp$ by the {\bf symplectic Fourier transform on $G$} of the characteristic function $\chi = \chi_\rho$, i.e.
			$$\W := \F_S(\chi),\;\; \W(\xv) = \int_G \chi(\yv) \overline{\Phi(\xv)(\yv)}d\mu(\yv),\; \xv\in G.$$
	\end{defn}

\begin{rem}
Using the twisted and the ordinary Plancherel theorems on $G$ (\eqref{eq-twisted-Plancherel} and \eqref{eq-Plancherel}) we can easily see that the characteristic/Wigner functions $\chi_\rho$ and $\W_\rho$ are well-defined as $L^2$-functions on $G$ for $\rho \in \mc S^2(\Hi)$.
\end{rem}

\begin{prop}\label{prop-Wigner-properties}
For a state $\rho \in \mc D$ we have the following.
    \begin{enumerate}
        \item When the Wigner function $\W_\rho$ is integrable on $G$ we have
            $$\int_G \W_\rho(\yv)d\widehat{\mu}(\yv) = 1 = \chi_\rho(0),$$
        where $\widehat{\mu}$ is the dual Haar measure respecting the Plancherel theorem \eqref{eq-Plancherel}. In general, we still have $\chi_\rho(0) = 1$.
        
        \item The Wigner function $\W_\rho$ is real-valued when the 2-cocycle $\sigma$ is {\bf normalized}, i.e. $\sigma(\xv,-\xv)=1$, $\xv\in G$.
    \end{enumerate}
\end{prop}

Another important consequence of the twisted and the ordinary Plancherel theorems on $G$ (\eqref{eq-twisted-Plancherel} and \eqref{eq-Plancherel}) is the following.

	\begin{thm} \label{thm-Plancherel-combine}
	Suppose that the 2-cocycle $\sigma$ is normalized. Then $\rho \in \mc D$ and $A = A^*\in S^2(\Hi)$ we have
		\begin{equation}\label{eq-expectations}
		    {\rm Tr}(\rho A) = \int_G \W_\rho \W_A d\widehat{\mu}.
		\end{equation}
	In other words, the ``{\bf quantum expectation}" ${\rm Tr}(\rho A)$ of a quantum observable $A$ w.r.t. the state $\rho$ is the same as the ``{\bf classical expectation}" $\int_G\W_\rho \W_A d\mu$ of $\W_A$ with respect to a real-valued normalized function $\W_\rho$.
	
	In particular, if $\rho$ is a state in $\D(\Hi)_{\W\ge 0} := \{\rho \in \mathcal{D}(\Hi): \W_\rho \ge 0\}$, then the ``classical expectation" $\int_G\W_\rho \W_A d\mu$ actually becomes a genuine probabilistic expectation.
	\end{thm}

\begin{rem}
    \begin{enumerate}
        \item We may take $A\in \mc B(\Hi)$ in \eqref{eq-expectations} by restricting the choice of states $\rho$ to a smaller class than $\mc D$ in infinite dimensional cases from Example \ref{ex-bicharacter}. See Proposition \ref{prop-extended-Plancherel} for the details. \color{black}
        
        \item The class $\D(\Hi)_{\W\ge0}$ was highlighted in the following result of Hudson~\cite{hudson1974wigner} and Soto/Claverie~\cite{soto1983wigner}: for a pure $n$-mode bosonic quantum state $\rho$ it is a bosonic gaussian state if and only if $\rho \in \D(\Hi)_{\W\ge0}$.
    \end{enumerate}
\end{rem}
	
The symmetry of the phase space is an important ingredient for the analysis of the bosonic systems. We have its abstract version as follows.

	\begin{defn}
	We say that a topological automorphism $S$ on $G$ is a {\bf symplectic map} (with respect to $\sigma$) if it is $\sigma$-preserving, i.e. $\sigma(S\xv, S\yv) = \sigma(\xv,\yv)$, $\xv,\yv\in G$). The group of all symplectic maps on $G$ with respect to $\sigma$ will be denoted by $Sp(G,\sigma)$, which we call the {\bf symplectic group} on $(G,\sigma)$. We say that a unitary $U \in \B(\Hi)$ is a {\bf gaussian unitary} if there is a symplectic map $S$ on $G$ such that
		\begin{equation}\label{eq-gaussian-unitary}
		UW(\xv)U^* = W(S\xv),\;\;\xv\in G.
		\end{equation}
	We denote $U$ by $U_S$ to emphasize the connection between $U$ and $S$.
	\end{defn}
	
\begin{rem} We will see in Proposition \ref{prop-symplectic-group} that $Sp(\Real^{2n},\sigma_{\rm boson})=Sp_{2n}(\Real) :=\{S\in M_{2n}(\Real): S^T JS=J\}$, the usual symplectic group on $\Real^{2n}$. This justifies the term ``symplectic''.
\end{rem}
	
On the other hand, unitary conjugation with respect to Weyl operators are easy to describe as follows.
    $$W(\yv)W(\xv)W(\yv)^* = \Phi(\yv)(\xv)W(\xv),\;\; \xv, \yv\in G.$$
Combining the above two types of unitaries we have the following Clifford covariance of Wigner functions.
	\begin{thm}\label{thm-Clifford covariance}
	Let $U=W(\yv)U_S$ for some $\yv\in G$ and $S\in Sp(G,\sigma)$. Then, there is a constant $C_S>0$, depending only on $S$, such that we have
	    $$\W_\rho(\xv) = C_S\W_{U\rho U^*}(S\xv + \yv),\; \rho \in \D(\Hi),\; \xv \in G.$$
	\end{thm}

\begin{rem}
In all the concrete examples we consider in this paper we can check that $C_S\equiv 1$. See Proposition \ref{prop-symplectic-group}, Remark \ref{rem-symplectic-discrete}, and Remark \ref{rem-measure-preserving}
\end{rem}

	\begin{defn}\label{def-Clifford-operation} We  call the unitaries of the form $W(\yv)U_S$, $\yv\in G$, $S\in Sp(G,\sigma)$ by {\bf Clifford operations on $(G,\sigma)$}. The {\bf Clifford group} $\mathcal{C}(G,\sigma)$ is defined by
	    $$\{U \in \mc U(\Hi_W): \text{$U$ is a Clifford operation on $(G,\sigma)$} \}/\tor.$$
	\end{defn}

The above group is nothing but a semi-direct product of $G$ and $Sp(G,\sigma)$ as follows. See Remark \ref{rem-abs-Clifford-semidirect} for the details. 
\begin{prop}\label{prop-abs-Clifford-semidirect}
We have a topological group isomorphism $\mathcal{C}(G,\sigma) \cong G \rtimes Sp(G,\sigma)$.
\end{prop}

The class $\D(\Hi)_{\W\ge0}$ of all quantum states with non-negative Wigner functions (introduced in Theorem \ref{thm-Plancherel-combine}) is clearly preserved under the above symmetry.

	\begin{prop}
	Let $U$ be a Clifford operation on $(G,\sigma)$. Then, we have
		$$U \mathcal{D}(\Hi)_{\W\ge 0} U^* = \mathcal{D}(\Hi)_{\W\ge 0}.$$
	\end{prop}

A Clifford operation $U=W(\yv)U_S$ on $(G,\sigma)$ satisfies $UW(\xv)U^*=\xi(\xv)W(S(\xv))$ for some character
\begin{equation} \label{eq-Clifford-character}
    \xi=\Phi(\yv)(S\,\cdot)=\Phi(S^{-1}\yv)(\cdot) \in \widehat{G},
\end{equation}
depending on $\yv \in G$ and $S\in Sp(G,\sigma)$. We may ask whether this is the only possibility for symmetry preserving transformations, which allows us the following extended definition.
    \begin{defn}
    We call a unitary $U \in \mc B(\Hi_W)$ a {\bf generalized Clifford operation on $(G,\sigma)$} if there is a continuous map $S:G\to G$ and a Borel map $\xi:G\to\tor$ such that
    \begin{equation} \label{eq-Clifford}
        UW(\xv)U^*=\xi(\xv)W(S\xv),\;\; \xv\in G.
    \end{equation}
    The {\bf generalized Clifford group} $\mathcal{C}_{\rm gen}(G,\sigma)$ is defined by
	    $$\{U \in \mc U(\Hi_W): \text{$U$ is a generalized Clifford operation on $(G,\sigma)$}\}/\tor.$$    
    \end{defn}

Note that we do not assume any additivity or bijectivity conditions on $S$ in the definition of generalized Clifford groups. Nevertheless, it can be shown that $S$ is actually a monomorphism and becomes isomorphic in many cases (Proposition \ref{prop-Clifford-injectivity}). We also have collected several other properties of generalized Clifford groups with proofs in Section \ref{subsec-symmetry-groups}.  Furthermore, we have a general principle describing when generalized Clifford operations are indeed Clifford operations. 

\begin{prop}
\label{prop-abs-Clifford}
     Let $U\in \mathcal{U}(\Hi)$ be a generalized Clifford operation on $(G, \sigma)$ with the associated maps $\xi$ and $S$. Then the following are equivalent:
    \begin{enumerate}
        \item $U$ is a Clifford operation,
        \item $S\in Sp(G,\sigma)$,
        \item $\xi\in \widehat{G}$.
    \end{enumerate}
\end{prop} 

It is natural to ask whether we could determine the objects $Sp(G,\sigma)$ and $\mathcal{C}_{\rm gen}(G,\sigma)$ and whether the inclusion $\mathcal{C}(G,\sigma) \subseteq \mathcal{C}_{\rm gen}(G,\sigma)$ is proper or not, which will be examined for detailed examples in the later sections.

\section{Abstract Weyl systems: first examples}\label{sec-first-ex}


Here, we present examples of the abstract phase space $G$ and a 2-cocycle $\sigma$ on it of the form
$G=F\times \widehat{F}$ for another locally compact abelian group $F$ called the ``{\em configuration space}''. Note that $\widehat{G} = \widehat{F}\times F \cong G$ via the swap. We will consider the following {\em canonical} choice of 2-cocycle, $\sigma_{\rm can}: G\times G \to \tor$ given by
	$$\sigma_{\rm can}((x,\gamma), (x', \gamma')) :=  \gamma(x'),\; x,x'\in F,\; \gamma, \gamma'\in \widehat{G}.$$
In this case, the unique irreducible $\sigma_{\rm can}$-representation $W = W_{\sigma_{\rm can}}$ can be described as follows (\cite{Prasad11}).  We first define the translation operator $T_x$ and the modulation operator $M_\gamma$ for $x\in F$ and $\gamma \in \widehat{F}$ acting on $\Hi_W = L^2(F)$ by
	$$T_xf(u) := f(u-x),\;\; M_\gamma f(u) = \gamma(u)f(u),\; f\in L^2(F),\; u\in F.$$
Then, $W:G \to \mc B(\Hi_W)$ is given by
	$$W(x,\gamma) := T_xM_\gamma,\; (x,\gamma)\in G.$$
	
The above 2-cocycle $\sigma_{\rm can}$ is never {\em normalized} except the trivial group case. However, there is the {\em canonical normalization} $\tilde{\sigma}_{\rm can}$ given by
		\begin{equation}\label{eq-normalization}
		\tilde{\sigma}_{\rm can}(\xv,\yv) := \frac{\xi(\xv)\xi(\yv)}{\xi(\xv+\yv)} \sigma_{\rm can}(\xv,\yv)
		\end{equation}
with $\xi(\xv) = \overline{\sigma_{\rm can}(\xv,-\xv)^{\frac{1}{2}}},\; \xv,\yv\in G.$ The existence of such a Borel function $\xi: G \to \tor$ satisfying the equation \eqref{eq-normalization} means us that $\sigma_{\rm can}$ and $\tilde{\sigma}_{\rm can}$ are {\em similar} as 2-cocycles. Now the unique irreducible $\tilde{\sigma}_{\rm can}$-representation $\widetilde{W} = W_{\tilde{\sigma}_{\rm can}}$ becomes
	    $$\widetilde{W}(\xv) = \xi(\xv)W(\xv),\; \xv\in G.$$
	
Note finally that it is straightforward to check that both of $\sigma_{\rm can}$ and $\tilde{\sigma}_{\rm can}$ are Heisenberg multipliers.

\begin{rem}
In the above the choice of the function $\xi$ (so that the choice of $\tilde{\sigma}_{\rm can}$) actually depends on the choice of square roots, where multiple choices are possible. In all the examples in this paper we will specify the choice of $\xi$, which we might be able to call it ``canonical".
\end{rem}

\begin{ex}\label{ex-bicharacter}
\begin{enumerate}
	\item ({\bf Bosonic system in $n$-modes}) 
	For $F = \Real^n$ we identify $G=\Real^n \times \widehat{\Real^n} \cong \Real^{2n}$ via the map $(x,\gamma_p) \mapsto (x, p)$, where $\gamma_p\in \widehat{\Real^n}$ is given by $\gamma_p(x) = \exp(ix^T p),\; x,p\in \Real^n$. The pair $(G, \tilde{\sigma}_{\rm can})$ gives rise to the {\bf $n$-mode bosonic system} and we recover \eqref{eq-Euclid-2-cocycle}, namely $\tilde{\sigma}_{\rm can}(\xv, \yv) = \sigma_{\rm boson}(\xv, \yv) =  \exp(-\frac{i}{2}\xv^T J \yv), \;\;\xv, \yv \in G$ by taking $\xi(\xv)=\xi(x,p):=\exp(\frac{i}{2}x^T p)$ in \eqref{eq-normalization}. Our choice of Haar measures $\mu$ and $\widehat{\mu}$ on $G$ respecting \eqref{eq-twisted-Plancherel} and \eqref{eq-Plancherel} are $d\mu(\xv) = d\widehat{\mu}(\xv) = \frac{d\xv}{(2\pi)^n}$.
\color{black}
			
	\item ({\bf Angle-number system}) For $F = \tor$ with the identification $\z \cong \widehat{F}$ via $n\mapsto \gamma_n$ given by $\gamma_n(e^{2\pi i\theta}) = e^{2\pi i n\theta}$, $\theta \in [0,1)$. The pair $(G \cong \tor \times \z, \sigma_{\rm can})$ describes the {\bf angle-number system} from \cite{Wer-preprint, Wer16}. 
    The canonical 2-cocycle $\sigma_{\rm can}$ is given by
	    $$\sigma_{\rm can}((\theta, n), (\theta', n')) = e^{2\pi i n\theta'},\, \theta, \theta'\in [0,1), n, n'\in \z.$$
	Unlike the bosonic systems, there is no continuous $\xi$ such that $\xi(\theta, n)^2=e^{2\pi in\theta}, \, (\theta, n)\in \tor\times\z$. One of the natural choice (with some discontinuities) would be
	    $$\xi(\theta,n):=e^{\pi in\{\theta\}},\, (\theta, n)\in \tor\times\z$$
	where $\{x\}:=x-\lfloor x\rfloor$ denotes the fractional part of $x\in \Real$. Then the normalization $\tilde{\sigma}_{\rm can}$ of $\sigma_{\rm can}$ is computed as in \eqref{eq-normalization}, which may not be written in simple formula, though.
	
	Our choice of Haar measures $\mu$ and $\widehat{\mu}$ on $G$ respecting \eqref{eq-twisted-Plancherel} and \eqref{eq-Plancherel} are given by
	    $$\int_{G}f(\theta, n)d\mu(\theta, n) = \int_{G}f(\theta, n)d\widehat{\mu}(\theta, n) = \sum_{n\in \z}\int_\tor f(\theta, n)d\theta$$
	for $f(\cdot, n) \in C(\tor)$ and $f(\cdot, n) \equiv 0$ except for finitely many $n\in \z$.
	
	Let us examine the Wigner function with respect to $\tilde{\sigma}_{\rm can}$. We begin with the characteristic function $\chi_{|m\ra \la m'|}$ with respect to $\sigma_{\rm can}$ for the rank 1 operator $|m\ra \la m'|$, $m,m'\in \z$, where $|m\ra$ refers to the $m$-th element of the canonical orthonormal basis $\{e_m\}_{m\in \z}$ for $L^2(\tor)$ given by $e_m(\theta)=e^{2\pi i m\theta}$, $\theta\in \tor$. It is straightforward to check that $\chi_{|m\ra \la m'|}(\theta, n) = e_m(\theta) \delta_{m,m'+n}$, so that we have $\tilde{\chi}_{|m\ra \la m'|}(\theta, n) = e^{-\pi i n\theta}e_m(\theta) \delta_{m,m'+n}$ for $\theta\in \tor$, $n\in \z$, where $\tilde{\chi}_{|m\ra \la m'|}$ is the with respect to the characteristic function with respect to $\tilde{\sigma}_{\rm can}$. Taking symplectic Fourier transform we get $\widetilde{\W}_{|m\ra \la m'|}(\theta, n) = e_{m-m'}(\theta)\int_\tor e^{\pi i (m+m')\alpha}e^{-2\pi i n \alpha}d\alpha$, the Wigner function with respect to $\tilde{\sigma}_{\rm can}$. For an arbitrary state $|\psi \ra \in L^2(\tor)$ the above can be written as
	    $$\widetilde{\W}_{|\psi\ra \la \psi|}(\theta, n) = \int_\tor \psi(\theta + \frac{\alpha}{2})\overline{\psi(\theta - \frac{\alpha}{2})}e^{-2\pi i n \alpha}d\alpha,$$
	which is the same as the existing definition by Mukunda (\cite{Mukunda1979}) upto a universal constant. Note that the convention in \cite{Mukunda1979} uses $[-\frac{1}{2}, \frac{1}{2})$ instead of our choice $[0,1)$.

	\item ({\bf Finite Weyl system})
	For $F = \z^n_d$ we identify $G=\z^n_d \times \widehat{\z^n_d} \cong \z^{2n}_d$ via the map $(x,\gamma_p) \mapsto (x, p)$, where $\gamma_p\in \widehat{\z^n_d}$ is given by $\gamma_p(x) = \om(x^T p),\; x,p\in \z^n_d$. 
    The pair $(G, \sigma_{\rm can})$ was used to describe {\bf ``quantum systems with a finite number of states''} in \cite{CCS}. In particular, if $d\ge 3$ is odd, then we can take canonical $\xi$ to be $\xi(x,p):=\om(2^{-1}x^Tp)$, where $2^{-1} = \frac{d+1}{2}$ is the multiplicative inverse of $2$ in the ring $\z_d$. Therefore, the canonical normalization $\tilde{\sigma}_{\rm can}$ has a simple formula as follows:
		$$\tilde{\sigma}_{\rm can}(\xv, \yv) = \om^{-2^{-1}\xv^TJ \yv}, \xv, \yv \in \z^{2n}_d,$$
	where $J={\footnotesize \begin{bmatrix}0 & I_n\\ -I_n & 0\end{bmatrix}}\in M_{2n}(\z_d)$. The pair $(G,\tilde{\sigma}_{\rm can})$ was used to describe {\bf ``finite dimensional quantum system''} in \cite{Gro}.
	
	Our choice of Haar measures $\mu$ and $\widehat{\mu}$ on $G$ respecting \eqref{eq-twisted-Plancherel} and \eqref{eq-Plancherel} are $\mu(A) = \widehat{\mu}(A) = \frac{|A|}{d^n}$ for $A\subseteq G$.
				
\end{enumerate}	

\end{ex}

The process of canonical normalization allows the abstract symplectic group to be larger in some cases. 

\begin{prop} \label{prop-symmetry-group-inclusion}
Suppose that we may choose a Borel function $\xi(\xv) = \overline{\sigma(\xv,-\xv)^{\frac{1}{2}}}$, $\xv\in G$, satisfying $\xi\circ S = \xi$  for any $S\in Sp(G,\sigma)$. Then for the canonical normalization $\tilde{\sigma}$ of the 2-cocyle $\sigma$ given by $\tilde{\sigma}(\xv,\yv) = \frac{\xi(\xv)\xi(\yv)}{\xi(\xv+\yv)}\sigma(\xv,\yv)$, $\xv,\yv\in G$, we have $Sp(G,\sigma) \subseteq Sp(G,\tilde{\sigma})$ and therefore $\mathcal{C}(G, \sigma) \subseteq \mathcal{C}(G, \tilde{\sigma})$.
\end{prop}

We end this subsection with a table summarizing calculations on the symmetry groups $Sp(G,\sigma)$ and $\mathcal{C}_{\rm gen}(G,\sigma)$ for the abstract Weyl systems. Here, we use the symbol $L={\footnotesize \begin{bmatrix} 0&0\\I_n&0\end{bmatrix}}$ and note that the last column for the following table focuses on the issue whether we could identify the generalized Clifford group $\mathcal{C}_{\rm gen}(G,\sigma)$ with any of the Clifford group. See Section \ref{subsec-symmetry-groups} for the details.
		
\vspace{0.5cm}
		
\begin{table}[ht!]
  \begin{center}
    \caption{Symmetry groups for abstract Weyl systems}
    \label{tab:table1}
    \begin{tabular}{l|c|c}
    $(G,\sigma)$& $Sp(G,\sigma)$ & $\mathcal{C}_{\rm gen}(G,\sigma)$\\
    \hline
    $(\Real^{2n}, \sigma_{\rm can})$ &  $\{S\in M_{2n}(\Real): S^T L S = L\}$ & $=\mc C(\Real^{2n}, \tilde{\sigma}_{\rm can})$\\
    \hline
    $(\Real^{2n}, \tilde{\sigma}_{\rm can})$ & $Sp_{2n}(\Real)$ & $=\mc C(\Real^{2n}, \tilde{\sigma}_{\rm can})$\\
    \hline
    $(\tor\times \z, \sigma_{\rm can})$ & $\{\pm id_{\tor\times\z}\}\cong \z_2$ & $\neq\mc C(\tor \times \z, \sigma_{\rm can})$\\
    \hline
    $(\tor\times \z, \tilde{\sigma}_{\rm can})$ & $\{id_{\tor\times\z}\}$ & $\neq\mc C(\tor \times \z, \sigma_{\rm can})$\\
    \hline
    $(\z^{2n}_d, \sigma_{\rm can})$ odd integer $d\geq 3$ & $\{S\in M_{2n}(\z_d): S^T L S = L\}$ & $=\mc C(\z^{2n}_d, \tilde{\sigma}_{\rm can})$\\
    \hline
    $(\z^{2n}_d, \tilde{\sigma}_{\rm can})$ odd integer $d\ge 3$ & $Sp_{2n}(\z_d)$ & $=\mc C(\z^{2n}_d, \tilde{\sigma}_{\rm can})$
    \end{tabular}
  \end{center}
\end{table}

\section{Fermions, hardcore-bosons and more: second examples}\label{sec-second-ex}

Some of the methods and results associated with the Weyl-Wigner representation of systems of bosons already have a powerful analogue for fermions. Whereas the Wigner functions for state and observables of bosonic systems are elements of a commutative algebra (the group algebra in our formalism), the fermionic analogues live in a non-commutative Grassmann algebra~\cite{cahill1999density,bravyi2004lagrangian}.

This approach does not yield a quasi-probability representation. Instead, tractable solutions are obtained solely through the concept of Gaussian states and Gaussian unitaries (defined below). 

In this section we will provide an alternative formalism for Fermionic systems in terms of actual real functions over a ``phase-space'' using the Weyl-Wigner representation we developed. This approach can easily be extended to the case of mixed spin systems with a particular case of hard-core bosons.

\subsection{Fermions and associated gaussian states}

The $n$-mode fermionic system is described by the {\em Majorana operators} $\ch_1 , \dots, \ch_{2n}$, which are self-adjoint operators acting on $\Hi = \Comp^{2^n}= \ell^2(\z^n_2)$ satisfying the CAR: 
    $$\{\ch_j, \ch_k\} = 2\delta_{jk},\; 1\le j,k\le 2n.$$
Note that $\ch_j$'s are identified with
\begin{align*}
    \ch_{2j-1}=Z\otimes\cdots\otimes Z\otimes X\otimes I\otimes \cdots \otimes I\\
    \ch_{2j}=Z\otimes\cdots\otimes Z\otimes Y\otimes I\otimes \cdots \otimes I,
\end{align*}
where
    $$X={\footnotesize \begin{bmatrix} 0&1\\1&0\end{bmatrix}},\,Y={\footnotesize \begin{bmatrix} 0&-i\\i&0\end{bmatrix}},\,Z={\footnotesize\begin{bmatrix} 1&0\\0&-1\end{bmatrix}},$$
the usual Pauli matrices for qubit, and we have $X$ and $Y$ at $j$-th tensor component in the above.

We call a quantum state acting on $\Hi$ by an $n$-mode fermionic state. It is well known that any $n$-mode fermionic state (more generally any $2^n\times 2^n$ matrix) $\rho$ can be expressed as a polynomial of Majorana operators. More precisely, we have
    $$\rho = aI + \sum^{2n}_{k=1}\sum_{1\le j_1 < \cdots < j_k\le 2n}a_{j_1\cdots j_k}\ch_{j_1}\cdots\ch_{j_k}$$
for some complex numbers $a$ and $a_{j_1\cdots j_k}$'s.

One popular tool in the analysis of fermionic systems is Grassmann variables $\{\theta_j\}^{2n}_{j=1}$ satisfying $\theta^2_j = 0$ and $\theta_j\theta_k + \theta_k\theta_j = 0$ for $1\le j,k\le 2n$. We call the unital associative algebra $\G_n$ generated by $\{\theta_j\}^{2n}_{j=1}$ as the {\em Grassmann algebra}. We consider the following linear bijection:
	$$\om: M_{2^n}=\mc B(\Hi) \to \G_{2n},\;\; \ch_{j_1}\cdots \ch_{j_m}\mapsto \theta_{j_1}\cdots \theta_{j_m}$$
for $1\le j_1< \cdots < j_m \le 2n$. For $\rho \in M_{2^n}$ we call $\om(\rho)$ the {\em Grassmann representation} of $\rho$. An important class of fermionic states are defined by the Grassmann representation, namely fermionic gaussian states. We call $\rho \in M_{2^n}$ a {\em fermionic gaussian state in $n$-modes} if
    \begin{equation}\label{eq-def-fer-gaussian}
        \om(\rho) = \frac{1}{2^n}\exp(\frac{i}{2}\theta^T M \theta)
    \end{equation}
for some $M = -M^T\in M_{2n}(\Real)$ such that $M^TM\le I$, where $\theta = (\theta_1, \cdots, \theta_{2n})^T$. In this case, the matrix $M$ is called the {\em covariance matrix} of $\rho$ since we have
	$$M_{jk} = \frac{i}{2}{\rm Tr}(\rho[\ch_j, \ch_k]),\;\; 1\le j,k\le 2n.$$

\subsection{Weyl-Wigner representation for Fermions}

We begin with the abstract phase space $G = \z^{2n}_2 \cong \z_2^n\times \widehat{\z_2^n}$. The choice of 2-cocycles is different from (3) and (4) of Example \ref{ex-bicharacter} as follows.  We define 
    $$\sigma_{\rm fer}(\xv, \yv) := (-1)^{\xv^T \Delta \yv},\; \xv,\yv \in \z^{2n}_2, \;\text{where} \;\Delta = {\Tiny \begin{bmatrix} 0 &&&\\ 1& 0&&\\1&1&0&\\ \vdots& \vdots& \ddots & \ddots & \\1 & 1 & \cdots & 1 & 0 \end{bmatrix}}.$$
The pair $(\z^{2n}_2, \sigma_{\rm fer})$ describes the {\em $n$-mode fermionic system}. We can easily see that $\sigma_{\rm fer}$ is a Heisenberg multiplier (Lemma \ref{lem-fermion-Heisenberg}). The unique irreducible unitary $\sigma_{\rm fer}$-representation (Lemma \ref{lem-irreducibility-W-eps}) $W = W_{\rm fer}:\z^{2n}_2 \to \mc U(\ell^2(\z^n_2))=U(2^n)$ is given by $$W_{\rm fer}(\xv) := \ch^{x_1}_1 \cdots \ch^{x_{2n}}_{2n},\;\; \xv = (x_1,\cdots,x_{2n}) \in \z^{2n}_2.$$
We have an {\em ``equivalent description"} of the $n$-mode fermionic system with the canonical normalization $\tilde{\sigma}_{\rm fer}$ of $\sigma_{\rm fer}$ given by $\displaystyle \tilde{\sigma}_{\rm fer}(\xv,\yv) = \frac{\xi(\xv)\xi(\yv)}{\xi(\xv+\yv)} \sigma_{\rm fer}(\xv,\yv)$ with $\xi(\xv) = \overline{\sigma_{\rm fer}(\xv,-\xv)^{\frac{1}{2}}},\; \xv,\yv\in G.$
Then, the map $\widetilde{W} = \widetilde{W}_{\rm fer}:\z^{2n}_2 \to U(2^n),\; \xv \mapsto W_{\rm fer}(\xv)\xi(\xv)$ is the unique irreducible unitary $\tilde{\sigma}_{\rm fer}$-representation of $\z^{2n}_2$.

Now we have two versions of characteristic/Wigner functions for a $n$-mode fermionic state $\rho$, namely $(\chi, \W)$ and $(\tilde{\chi}, \widetilde{\W})$ w.r.t the 2-cocycles $\sigma_{\rm fer}$ and $\tilde{\sigma}_{\rm fer}$, respectively. More precisely, we have
    \begin{align}\label{eq-fermi-char-Wigner}
    \chi(\xv) &= {\rm Tr}(\rho\, W(\xv)^*),\;\; &\W(\yv) &= 2^{-n}\sum_{\xv \in \z^{2n}_2}(-1)^{\xv^T(\Delta + \Delta^T)\yv}\chi(\xv),\\
    \tilde{\chi}(\xv) &= {\rm Tr}(\rho\, \widetilde{W}(\xv)^*),\;\; &\widetilde{\W}(\yv) &= 2^{-n}\sum_{\xv \in \z^{2n}_2}(-1)^{\xv^T(\Delta + \Delta^T)\yv}\tilde{\chi}(\xv),\;\; \yv\in G. \nonumber
    \end{align}
Here, we are using the choice of Haar measures $\mu$ and $\widehat{\mu}$ on $G$ given by $\mu(A) = \widehat{\mu}(A) = \frac{|A|}{2^n}$ for $A\subseteq G$, which respect \eqref{eq-twisted-Plancherel} and \eqref{eq-Plancherel}.    

\subsection{Mixed spin systems, hardcore-bosons and associated symmetry groups}\label{subsec-mixed-spin}

We can generalize the Fermionic system according to the following arbitrary choice of signs
    $$\eps: \{1,\cdots, n\} \times \{1,\cdots, n\} \to  \{\pm 1\}$$
satisfying (i) $\eps(i,i)=-1$, (ii) $\eps(i,j) = \eps(j,i)$, $1\le i, j \le n$. The {\em $n$-mode mixed spin system} is described by the {\em $\eps$-Majorana operators} $\ch_{\eps,1} , \dots, \ch_{\eps,2n}$, which are self-adjoint operators acting on $\Hi = \Comp^{2^n}= \ell^2(\z^n_2)$ satisfying the $\eps$-CAR:     $$\begin{cases}\ch_{\eps,2j}\ch_{\eps,2k} - \eps(j,k) \ch_{\eps,2j}\ch_{\eps,2k} = 2\delta_{jk}\\
\ch_{\eps,2j-1}\ch_{\eps,2k-1} - \eps(j,k) \ch_{\eps,2j-1}\ch_{\eps,2k-1} = 2\delta_{jk}\\
\ch_{\eps,2j-1}\ch_{\eps,2k} =  \eps(j,k) \ch_{\eps,2j-1}\ch_{\eps,2k}\end{cases} \text{for $1\le j,k\le n$}.$$
Note that $\ch_{\eps,j}$'s are identified with
\begin{align*}
    \ch_{\eps, 2j-1}=Z_{\eps(1,j)}\otimes\cdots\otimes Z_{\eps(j-1,j)}\otimes X\otimes I\otimes \cdots \otimes I\\
    \ch_{\eps, 2j}=Z_{\eps(1,j)}\otimes\cdots\otimes Z_{\eps(j-1,j)}\otimes Y\otimes I\otimes \cdots \otimes I,
\end{align*}
where $Z_1 = I_2, Z_{-1} = Z\in M_2(\Comp)$. See \cite{Bi} for the details and the connection to free probability.


We use the same abstract phase space $G = \z^{2n}_2 \cong \z_2^n\times \widehat{\z_2^n}$ for the $n$-mode mixed spin system with the choice of 2-cocycle $\sigma_\eps$ given by $\sigma_\eps(\xv, \yv) := (-1)^{\xv^T \Delta_\eps \yv},\; \xv,\yv \in \z^{2n}_2,$\; \text{where}\;     $$\Delta_\eps = {\Tiny \begin{bmatrix} 0 &&&&&&\\ 
    1& 0&&&&\\
    \tilde{\eps}(1,2)&\tilde{\eps}(1,2)&0&&&&\\
    \tilde{\eps}(1,2)&\tilde{\eps}(1,2)&1&0&&&\\
    \vdots& \vdots& \vdots & \vdots &\ddots&&& \\
    \vdots& \vdots& \vdots & \vdots &\ddots&0&& \\
    \tilde{\eps}(1,n) & \tilde{\eps}(1,n) & \tilde{\eps}(2,n) & \tilde{\eps}(2,n) & \cdots &\tilde{\eps}(n-1,n)& 0 & \\
    \tilde{\eps}(1,n) & \tilde{\eps}(1,n) & \tilde{\eps}(2,n) & \tilde{\eps}(2,n) & \cdots &\tilde{\eps}(n-1,n) & 1 & 0\end{bmatrix}}\in M_{2n}(\z_2)$$
and $\tilde{\eps}(i,j) = \frac{1-\eps(i,j)}{2}\in \z_2$, $1\le i,j\le n.$ The pair $(\z^{2n}_2, \sigma_\eps)$ describes the {\em $n$-mode mixed spin system}. It is not difficult to check that $\sigma_\eps$ is a Heisenberg multiplier for any choice of $\eps$ (Lemma \ref{lem-fermion-Heisenberg}). The unique irreducible unitary $\sigma_{\eps}$-representation (Lemma \ref{lem-irreducibility-W-eps}) $W_\eps:\z^{2n}_2 \to U(2^n)$ is given by
    \begin{equation}\label{eq-W-eps}
    W_\eps(\xv) := \ch^{x_1}_{\eps,1} \cdots \ch^{x_{2n}}_{\eps,2n},\;\; \xv = (x_1,\cdots,x_{2n}) \in \z^{2n}_2.    
    \end{equation}
Note that we recover the fermionic case when $\eps \equiv -1$, i.e. $\sigma_{-1} = \sigma_{\rm fer}$. When $\eps(i,j) = 1$ for all $1\le i \neq j\le n$ the associated quantum system corresponds to {\em hard-core bosons}, in the sense that the $\eps$-Majorana operators $\ch_{\eps,j}$
behave like Majorana operators, but commute for different modes~\cite{matsubara1956}.

Note finally that we also have canonical normalization $\tilde{\sigma}_{\eps}$ and the associated projective representation $\widetilde{W}_{\eps}$ as before.

\begin{rem}
When $n=1$ (i.e. the 1-mode case) we also have $\sigma_\eps = \sigma_{\rm fer} = \sigma_{\rm can}$ regardless of the choice of signs $\eps$.
\end{rem}

We summarize calculations on the symmetry groups $Sp(G,\sigma)$ and $\mathcal{C}_{\rm gen}(G,\sigma)$ for mixed spin systems in Table ~2. Here, $\mc C_n$ is the original {\em Clifford group on $n$-qubits} (by Gottesman \cite[pp. 127--128]{Got}) given by
    $$\mc C_n:=\{U\in U(2^n): UP_nU^*\subset \pm P_n\}/\tor,$$
where $P_n:=\{A_1\otimes\cdots\otimes A_n\, |\, A_j\in\{I, X, Y, Z\}\}$ for the Pauli matrices $X,Y,Z$ in the single qubit system.

\vspace{0.5cm}

\begin{table}[ht!]
  \begin{center}
    \caption{Symmetry groups for fermions and mixed spin systems}
    \label{tab:table2}
    \begin{tabular}{l|c|c}
    $(G,\sigma)$& $Sp(G,\sigma)$ & $\mathcal{C}_{\rm gen}(G,\sigma)$\\
    \hline
    $(\z^{2n}_2, \sigma_{\rm fer})$ & $\{S\in M_{2n}(\z_2)\, | \, S^T\Delta S=\Delta\}$ & $\mc C_n$\\
    \hline
    $(\z^{2n}_2, \tilde{\sigma}_{\rm fer})$ & $\z_3$ for $n=1$ (open for other cases) & $\mc C_n$\\
    \hline
    $(\z^{2n}_2, \sigma_\eps)$ & $\{S\in M_{2n}(\z_2)\, | \, S^T\Delta_\eps S=\Delta_\eps\}$ & $\mc C_n$\\
    \hline
    $(\z^{2n}_2, \tilde{\sigma}_\eps)$ & ? (open) & $\mc C_n$
    \end{tabular}
  \end{center}
\end{table}

\begin{rem}
There is a generalized Clifford operation on $(\mathbb{Z}_2^2, \sigma_{\rm fer})$ and $(\z_2^2,\tilde{\sigma}_{\rm fer})$, which is not a Clifford operation, i.e. 
    $$\mc C(\mathbb{Z}_2^2, \sigma_{\rm fer}) \subsetneq \mc C(\z_2^2, \tilde{\sigma}_{\rm fer}) \subsetneq \mc C_2=\mc C_{\rm gen}(\mathbb{Z}_2^2, \sigma_{\rm fer})=\mc C_{\rm gen}(\mathbb{Z}_2^2, \tilde{\sigma}_{\rm fer}).$$
Indeed, all the inclusion relations except for $\mc C(\z_2^2, \tilde{\sigma}_{\rm fer}) \neq \mc C_2$ follow from Proposition \ref{prop-low-rank-symplectic} and Proposition \ref{prop-Clifford-fermion}. Note that we can find a generalized Clifford operation $U$ on $(\mathbb{Z}_2^2, \sigma_{\rm fer})$ with the following pair $(S,\xi)$:

\begin{center}		
$S = {\footnotesize \begin{bmatrix}0 & 1 \\ 1 & 0\end{bmatrix}}$;\quad
    {\footnotesize \begin{tabular}{c|cccc}
    $\xv$ & $00$ & $10$ & $01$ & $11$ \\
    \hline $\xi(\xv)$& $1$ & $1$ & $1$ & $-1$
    \end{tabular}},
\end{center}
just by checking $\xv \mapsto \xi(\xv)W_{\rm fer}(S\xv)$ is a $\sigma_{\rm fer}$-representation of $\z_2^2$. However, $U$ is not a Clifford operation on $(\z_2^2, \tilde{\sigma}_{\rm fer})$ since $S\notin Sp(\z_2^2, \tilde{\sigma}_{\rm fer})$. Therefore, $\mc C(\z_2^2, \tilde{\sigma}_{\rm fer}) \subsetneq \mc C_2$.

\end{rem}

\color{black}
\subsection{Non-negative Wigner functions of Fermionic states}
\label{section:positive-wigner-fermion}

Now we investigate non-negativity of $\W=\W_\rho$ and $\widetilde{\W}=\widetilde{\W}_\rho$ for a Fermionic state $\rho$ (as in \eqref{eq-fermi-char-Wigner}) starting from 1-mode case. Note that a $1$-mode fermionic state is nothing but a qubit state, so that it is $\rho = {\footnotesize \begin{bmatrix}a & b \\ c & d\end{bmatrix}}$, $a,b,c, d\in \Comp$ with $a,d \ge 0, a+d=1, c = \bar{b}, |b|^2\le ad$.

\begin{thm}\label{thm-1-mode-fermi-Wigner-unnormal} Let $\rho = {\footnotesize \begin{bmatrix}a & b \\ c & d\end{bmatrix}}$ be a $1$-mode fermionic state.
	\begin{enumerate}
		
		\item We have $\W_\rho \ge 0$ if and only if $a=d=\frac{1}{2}$, $b=\bar{c}$ with $-\frac{1}{2} \le {\rm Re}\, b\pm {\rm Im}\,b \le \frac{1}{2}.$
		
		\item We have $\W_\rho \ge 0$, $\rho$ being pure if and only if $a=d=\frac{1}{2}$, $b =\bar{c} \in \{\pm \frac{1}{2}, \pm \frac{1}{2}i\}$.
		
		\item The state $\rho$ is a fermionic gaussian state with $\W_\rho \ge 0$ if and only if $\rho = \frac{1}{2}I_2$.
	\end{enumerate}	
\end{thm}

With the normalization we get a wider range of states with non-negative Wigner functions. Here we specify the function $\xi(\xv)=\overline{\sigma_{\rm fer}(\xv,\xv)^{\frac{1}{2}}}$ on a single mode fermionic system in order to avoid any confusion:
\begin{center}
    {\footnotesize \begin{tabular}{c|cccc}
    $\xv$ & $00$ & $10$ & $01$ & $11$ \\
    \hline $\xi(\xv)$& $1$ & $1$ & $1$ & $-i$
    \end{tabular}}
\end{center}

\begin{thm}\label{thm-1-mode-fermi-Wigner-normal} Let $\rho$ be a $1$-mode fermionic state.
	\begin{enumerate}
		\item $\widetilde{\W}_\rho \ge 0$ if and only if
			\begin{equation}\label{eq-1-mode-normalized}
			\rho = {\footnotesize \begin{bmatrix}a & x+iy \\ x-iy & 1-a\end{bmatrix}} \;\;\text{with}\;\; \begin{cases}x^2+y^2 \le a(1-a),\\ |x \pm y|\le \frac{1}{2}(1 \mp (1-2a)) \end{cases}\;  a,x,y\in \Real.
			\end{equation}
		
		\item In the above case $\rho$ is a pure state if and only if $\begin{cases}x^2+y^2 = a(1-a),\\ |x \pm y|\le \frac{1}{2}(1 \mp (1-2a)).\end{cases}$
		
		\item We always have $\widetilde{\W}_\rho \ge 0$ for any fermionic gaussian state $\rho$.
	\end{enumerate}	
\end{thm}

For the 2-mode case we still have the following complete picture for the un-normalized case.

\begin{thm}\label{thm-2-mode-fermi-Wigner-unnormal} Let $\rho$ be a $2$-mode fermionic state.
	\begin{enumerate}
		\item $\W_\rho \ge 0$ if and only if $\rho$ is of the form
	\begin{equation}\label{eq-matrix-rho-2-mode}
	\rho = \frac{1}{4} {\footnotesize \begin{bmatrix} 1-a_5 & a_3 - ia_4 & a_1 - i a_2 & 0 \\ a_3 + ia_4 & 1 + a_5 & 0 & a_1 - ia_2\\ a_1 + ia_2& 0 & 1+a_5 & -a_3 + ia_4\\0 & a_1 + i a_2 & -a_3 - ia_4 & 1- a_5  \end{bmatrix}}
	\end{equation}
	with $(a_j)^5_{j=1} \subseteq [-1,1]$ and
				\begin{equation}\label{eq-2-mode-condition}\begin{cases} |(a_1+a_2) \pm (a_3+a_4)| \le 1+a_5\\ |(a_1-a_2) \pm (a_3-a_4)| \le 1+a_5\\ |(a_1+a_2) \pm (a_3-a_4)| \le 1-a_5\\ |(a_1-a_2) \pm (a_3+a_4)| \le 1-a_5.\end{cases}
				\end{equation}

		\item There is no pure state $\rho$ with $\W_\rho \ge 0$.
		
		\item The state $\rho$ in \eqref{eq-matrix-rho-2-mode} is fermionic gaussian if and only if $a_j=0$, $1\le j\le 5$.
	\end{enumerate}	
\end{thm}

The last result of the above theorem can be extended to higher dimensional cases.

\begin{thm}\label{thm-n-mode-fermi-Wigner-unnormal}
The maximally mixed state $2^{-n}I_{2^n}$ is the only $n$-mode fermionic gaussian state $\rho$ with $\W_\rho\ge 0$.
\end{thm}

The normalized case is far more complicated, so that we have restricted results as follows. 

\begin{thm}\label{thm-2-mode-fermi-Wigner-normal}
    \begin{enumerate}
        \item Let $\rho$ be a $n$-mode fermionic state with $\W_\rho \ge 0$. Then, we automatically have $\widetilde{\W}_\rho\ge 0$.
        
        \item Let $\rho$ be the $n$-mode fermionic gaussian state with the covariance matrix
    \begin{equation}\label{eq-correlation}
        M = \bigoplus^n_{j=1}{\footnotesize \begin{bmatrix}0&a_j\\-a_j&0\end{bmatrix}},\;\text{where}\; |a_j|\le 1, a_j\in \Real, 1\le j\le n.
    \end{equation} Then, $\widetilde{\W}_\rho \ge 0$ for $\prod^n_{j=1}(|a_j|+1) \le 2$.
    Moreover, there is a 2-mode pure fermionic gaussian state $\rho$ such that $\widetilde{\W}_\rho$ is not non-negative.
    \end{enumerate}
\end{thm}

\begin{rem}
We can see that the condition $\widetilde{\W}_\rho\ge 0$ covers wider range of states than the condition $\W_\rho\ge 0$, but still not enough to capture all of the fermionic gaussianity.
\end{rem}

\section{Outlook}

The tools we developed provide us with a large class of pseudo-probability distributions (Wigner functions) which generalizes the main known examples. 
For these representations to be useful in enabling classical simulations and solutions to the dynamics of quantum systems, it is important to have a good characterizations of both the states with non-negative Wigner functions, and quantum dynamics which preserve this property. 

We made progress on the later, by characterizing the group of quantum transformations whose effects on the distribution can be represented solely by a transformation on phase-space (Clifford operations). We have shown that this group is non-trivial also in the new examples that we proposed (fermions, hard-core bosons and angle-number system). However in the Fermionic case, we have only fully characterized the {\em generalized} Clifford group, which merely contain the group of interest.

However, apart from a few first results in Section~\ref{section:positive-wigner-fermion} and apart from the known results for bosons and finite Weyl systems, we do not know how to generally characterize the states with non-negative Wigner functions, i.e., how to generalize Hudson's theorem. A first simpler steps to be taken in that direction would include a formalization of the composition of subsystems and partial traces, and whether these operation preserve the positivity of the generalized Wigner functions. 

To build a more general framework, it is possible to drop the requirement that the 2-cocycle be a Heisenberg multiplier. This would cover cases where the quantum system is not a type I factor, or not a factor at all (corresponding to a non-unique projective representation). The case of a type II factor is considered for instance in Appendix~\ref{section:IRA}, but a full analysis is left for further work.



\color{black}

\subsection*{Acknowledgements} HHL and SJP's research was supported by the Basic Science Research Program through the National Research Foundation of Korea (NRF) Grant NRF-2017R1E1A1A03070510 and the National Research Foundation of Korea (NRF) Grant funded by the Korean Government (MSIT) (Grant No.2017R1A5A1015626).

\subsection*{Data Availability} Data sharing is not applicable to this article as no new data were created or analyzed in this study.

\appendix

\section{Twisted Fourier analysis}

\label{twistedfouriertransform}

In this section we collect mathematical aspects of (twisted) Fourier analysis on groups (\cite{Folland-book, KleppnerLipsman72}) including some of the proofs we need.

Let us begin with the abstract {\em phase space} $G$, which is a locally compact abelian group equipped with a Haar measure $\mu$, which is translation invariant. The choice of $\mu$ will be fixed later. In harmonic analysis we associate several algebras to the group $G$. The first such example would be $L^\infty(G)$, the algebra of essentially bounded functions on $G$, which is a commutative von Neumann algebra (via the usual element-wise multiplication of functions and complex conjugation). The space $L^1(G)$ consisting of $\mu$-integrable functions is the predual of $L^\infty(G)$ equipped with a commutative Banach $*$-algebra structure---the {\em group algebra} of $G$---with the the {\em convolution} product $*$
\begin{equation}
(f * g)(x) := \int_G f(x) g(y-x) d\mu(x),
\end{equation}
and the involution
\(
f^\star(x) = \overline f(-x),\; x\in G.
\)
The above convolution product $f*g$ can be extended to the case of $g\in L^2(G)$ via Young's inequality, so that we can define the convolution operator $L_f: L^2(G) \to L^2(G),\; g \mapsto f* g$. Then, we get another commutative von Neumann algebra ${\rm vN}(G) \subseteq \mc B(L^2(G))$, called the {\em group von Neumann algebra}, generated by convolution operators $L_f$, $f\in L^1(G)$.

Commutativity of the algebra ${\rm vN}(G)$ leads us to the Gelfand representation, ${\rm vN}(G) \simeq L^\infty(\widehat G)$, where  $\widehat G$ is the {\em dual group} of $G$. Recall that an element $\chi \in \widehat{G}$ is a continuous group homomorphism $\chi: G \rightarrow \mathbb T \cong U(1)$. In other words, it is an one-dimensional (thus, irreducible) unitary representation of $G$. The above identification ${\rm vN}(G) \simeq L^\infty(\widehat G)$ maps $L_f$ to $\F(f)$, where
    $$\F: L^1(G) \rightarrow L^\infty(\widehat G),\; f\mapsto \widehat{f}$$
is the Fourier transform on $G$ given by
    $$\widehat f(\gamma) := \int_G f(x) \overline{\gamma(x)} dx,\;\; \gamma\in \widehat{G}.$$
Note that Definition \ref{defn-Wignerftn} uses the assumption $G\cong \widehat{G}$ with the duality $\Phi: G \to \widehat{G}$ to get the symplectic Fourier transform $\F_S(f)$, which is in general different from the above $\widehat{f}$ using the canonical identification $x\in G \mapsto \gamma_x \in \widehat{G}$. The Fourier transform $\mc F$ is a $*$-algebra homomorphism from the group algebra $(L^1(G), *, \star)$ into the group von Neumann algebra ${\rm vN}(G)\subseteq \mc B(L^2(G))$. Moreover, the Fourier transform $\mc F$ can be extended to a unitary between corresponding $L^2$-spaces, i.e. there is a Haar measure $\widehat{\mu}$ on $\widehat{G}$ such that
    $$\mc F: L^2(G) \to L^2(\widehat{G}),\;\; f\mapsto \widehat{f}$$
is a unitary, which is called the {\em Plancherel theorem}. In particular, we have
    \begin{equation}\label{eq-Plancherel}
        \int_G f\bar{g} d\mu = \int_{\widehat{G}}\widehat{f}\, \bar{\widehat{g}}d\widehat{\mu}, \;\; f,g\in L^2(G).
    \end{equation}   

We might be able to say that the group von Neumann algebra ${\rm vN}(G)$ does not exhibit true quantum nature thanks to its commutativity. However, we may {\em twist} the algebra ${\rm vN}(G)$ via a {\em 2-cocycle}, which results a non-commutative algebra. The convolution product $*$ twists into a non-commutative product $*_\sigma$ (called the {\em twisted convolution}) defined by
\begin{equation}
(f *_\sigma g)(x) := \int_G f(y) g(x-y) \sigma(y,x-y) d\mu(y).
\end{equation}
We also have a twisted involution $f^{\star\sigma}(x) := \overline{\sigma(x,-x) f(-x)}$, $x\in G$. The above twisted convolution $f*_\sigma g$ can also be extended to the case of $g\in L^2(G)$ via Young's inequality, so that we can define the {\em twisted convolution operator} $L^\sigma_f: L^2(G) \to L^2(G),\; g \mapsto f*_\sigma g$, which in turn generates a non-commutative algebra called the {\em twisted group von Neumann algebra} ${\rm vN}(G,\sigma) \subseteq \mc B(L^2(G))$.

The story for {\em twisted Fourier transform} is a bit more involved, since we need to deal with possibly infinite dimensional irreducible {\em twisted unitary representation} of $G$. This is in stark contrast with the fact that every element $\gamma \in \widehat{G}$ is a 1-dimensional irreducible unitary representation of $G$ as mentioned before. In this paper we focus on the case when the 2-cocycle $\sigma$ is an {\em Heisenberg multiplier}, which forces that there is only one irreducible {\em $\sigma$-representation} $\weyl:  G \to \mathcal{U}(\mc H_\weyl)$ up to unitary equivalence. Now we can define the {\em twisted Fourier transform} $\F_\sigma$ by
	$$\F_\sigma: L^1(G) \to \B(\Hi_\weyl),\;\; f\mapsto \widehat{f}(W):= \int_G f(x) W(x) d\mu(x) \in B(\Hi_W),$$
which is still a $*$-algebra homomorphism from $(L^1(G), *_\sigma, \star_\sigma)$ into $\B(\Hi_\weyl)$. Moreover, the twisted Fourier transform $\mc F_\sigma$ can be extended to a unitary between corresponding $L^2$-spaces. For this result we need preparations. First, the space $G\times \tor$ with the group law $(x,z)\cdot (y,w) = (x+y, zw\sigma(x,y))$ becomes a locally compact group, which we call {\em the central extension $G(\sigma)$} of $G$. Secondly, we recall the {\em regular $\sigma$-representation} $\lambda_\sigma: G \to \mc B(L^2(G))$ given by
    \begin{equation}\label{eq-reg-rep}
        \lambda_\sigma(x)f(y) = \sigma(x,y-x)f(y-x),\; x,y\in G,\; f\in L^2(G).
    \end{equation}
\begin{thm}\label{thm-twistedPlancherel} ({\bf Twisted Plancherel theorem}, \cite[Theorem 7.1]{KleppnerLipsman72})
Suppose, in addition, that the central extension $G(\sigma)$ has a type I regular representation. Then the map
    $$\mc F_\sigma: L^2(G) \to \mc S^2(\Hi_W),\;\; f\mapsto \widehat{f}(W)$$
is a unitary. In particular, we have
    \begin{equation}\label{eq-twisted-Plancherel}
        \int_G f\bar{g} d\mu = {\rm Tr}(\widehat{f}(W) \widehat{g}(W)^*), \;\; f,g\in L^2(G)
    \end{equation}
for an appropriate choice of a Haar measure $\mu$ on $G$. Moreover, we have the unitary equivalence $\lambda_\sigma \cong W\otimes 1_W$ with the intertwiner $\F_\sigma$. Consequently, we have
    $${\rm vN}(G,\sigma) \simeq \mc B(\mc H_\weyl),\; L^\sigma_f = \int_G \lambda_\sigma(y)f(y) d\mu(y) \mapsto \widehat{f}(W).$$
\end{thm}

\begin{rem}\label{rem-typeI-condition}
For each of the examples in this paper the central extension $G(\sigma)$ is actually a type I group, so that the additional condition in Theorem \ref{thm-twistedPlancherel} is satisfied. Indeed, for the cases in Section \ref{sec-first-ex} we may appeal to \cite{Prasad11} (with a minor modification) and \cite[p. 207, Example 3]{Folland-book}. For the cases in Section \ref{sec-second-ex} we know that $G(\sigma)$ is compact, so that it is type I.
\end{rem}

\section{Characteristic and Wigner functions}

In this section we collect some of the essential properties of characteristic and Wigner functions.

\begin{prop}\label{prop-Fourier-char}
For any $\rho \in \mc S^1(\Hi)$ we have $\chi_\rho \in L^2(G)$ and $\F_\sigma(\chi_\rho) = \rho.$
\end{prop}
\begin{proof}
Since $\F_\sigma: L^2(G) \to \mc S^2(\Hi)$ is a unitary, span$\{\F_\sigma(\varphi): \varphi \in C_c(G)\}$ is dense in $S^2(\Hi)$. Here, $C_c(G)$ is the space of all continuous functions on $G$ with compact support. Consequently, span$\{\F_\sigma(\varphi_1)\F_\sigma(\varphi_2)^*: \varphi_1, \varphi_2 \in C_c(G)\}$ is dense in $S^1(\Hi)$. Now we have $\rho = \F_\sigma(\varphi_1)\F_\sigma(\varphi_2)^* = \F_\sigma(\varphi_1 *_\sigma \varphi^{\star_\sigma}_2)$ and
    $$\varphi_1 *_\sigma \varphi^{\star_\sigma}_2(\cdot) = \overline{\la \lambda_\sigma(\cdot)\varphi_2, \varphi_1\ra} = \textrm{Tr}(W^*(\cdot)\F_\sigma(\varphi_1)\F_\sigma(\varphi_2)^*) = \chi_\rho(\cdot).$$
Here, we used the fact that $\F_\sigma \circ \lambda_\sigma(\cdot) = (W(\cdot) \otimes 1_W)\circ \F_\sigma$.    
\end{proof}

Note that the integral $\int_G f(x) W(x) d\mu(x)$ is well-defined in the strong sense only for the case of $f\in L^1(G)$. However, we would like to understand its precise meaning for more general $f$.
    \begin{prop}\label{prop-op-integral-bdd}
    For $f\in L^2(G)$ the integral $\int_G f(\xv) W(\xv) d\mu(\xv)$ defines a bounded operator on $\Hi$ in the weak sense. Note that it actually is a Hilbert-Schmidt operator thanks to Theorem \ref{thm-twistedPlancherel}.
    \end{prop}
\begin{proof}
For any $\xi, \eta \in \Hi$ the rank 1 operator $|\xi\ra\la \eta|$ satisfies $\| |\xi\ra\la \eta| \|_1 = \| |\xi\ra\la \eta| \|_2 = \|\xi\|_2 \|\eta\|_2$, so that we have $\|\chi_{|\xi\ra\la \eta|}\|_2 = \|\xi\|_2 \|\eta\|_2$. Thus,
    \begin{align*}
        |\la\eta |\int_{G}f(\xv)W(\xv)d\mu(\xv)|\xi\ra|
        & = |\int_{G}f(\xv)\la\eta |W(\xv)|\xi\ra d\mu(\xv)|\\
        & = |\int_{G}f(\xv)\chi_{|\xi\ra\la \eta|} d\mu(\xv)|\\
        & \le \|f\|_2 \|\xi\|_2 \|\eta\|_2.
    \end{align*}
\end{proof}    

If the function $f$ on $G$ has enough regularity, then the integral $\int_G f(\xv) W(\xv) d\mu(\xv)$ even becomes trace class operators, which allows us to generalize Theorem \ref{thm-Plancherel-combine} in the case of $G = F\times \widehat{F}$ for $F = \Real^n$ or $\tor$ with the 2-cocycle $\sigma_{\rm can}$ as in Example \ref{ex-bicharacter}. For this choice of group we will consider the {\em ``Schwarz class"} $\mc S (G)$. For $G = \Real^{2n}$ we can take the usual Schwarz class, but for $G = \tor \times \z$, we will take
    $$\mc S (G) := \{f=(f_n)_{n\in \z}: f_n \in C^\infty(\tor),\; \sup_{n, m\in \z}(1+|nm|)^k|\widehat{f_n}^\tor(m)| < \infty,\; \forall k\in \n\},$$
where $\widehat{g}^\tor(n)$ refers to the $n$-th Fourier coefficient of a function $g$ on $\tor$. Note that we have $f(\theta, n) = f_n(\theta)$, $n\in \z$, $\theta \in [0,1]$ and it is relatively easy to see that the space $\mc S (G)$ is invariant under the Fourier transform on $G$. Indeed, we can easily see that for any $f=(f_n)_{n\in \z} \in \mc S(G)$ we have $\widehat{\F^G(f)_m}^\tor (n) = \widehat{f_n}^\tor(m)$, where $\F^G$ means the Fourier transform on $G$. The ``Schwarz class" $\mc S (G)$ is a locally convex topological vector space with the canonical topology, and we call the topological dual $\mc S(G)^*$ the space of all {\em tempered distributions on $G$} following the Euclidean case.
    \begin{prop}\label{prop-extended-Plancherel}
    Let $G = F\times \widehat{F}$ for $F = \Real^n$ or $\tor$ with the 2-cocycle $\sigma_{\rm can}$ as in Example \ref{ex-bicharacter}. Then, for any $f\in \mc S(G)$ the integral $\int_G f(\xv) W(\xv) d\mu(\xv)$ is a trace class operator. Moreover, for any $A=A^* \in \mc B(\Hi)$ the Wigner function $\W_A$ is well-defined as a tempered distribution on $G$ via
        $$\la \W_A, \varphi \ra := \text{\rm Tr}(A\, \F_\sigma(\F^{-1}_S(\varphi))),\;\; \varphi \in \mc S(G).$$
    The above replaces \eqref{eq-expectations} for the choice of $\rho = \F_\sigma(\F^{-1}_S(\varphi)) \in \F_\sigma(\mc S (G))$. This means that if we take the quantum observable $A\in \mc B(\Hi)$, which is a bigger class than $\mc S^2(\Hi)$ then the coincidence of the quantum and the classical expectations remains to be true if we restrict our choice of states in $\F_\sigma(\mc S (G))$, a smaller class than $\mc D$.
    \end{prop}
\begin{proof}
From the definition of $W(y,\gamma)$ we can readily see that for any $h\in L^2(F)$ and $u\in F$ we have
    \begin{align*}
        \lefteqn{\int_G f(y,\gamma)W(y,\gamma)d\mu_F(y)d\mu_{\widehat{F}}(\gamma)h(u)}\\
        & = \int_F \int_{\widehat{F}} f(y,\gamma)\gamma(u-y)h(u-y)d\mu_F(y)d\mu_{\widehat{F}}(\gamma)\\
        & = \int_F \int_{\widehat{F}} f(u-y,\gamma)\gamma(y)h(y)d\mu_{\widehat{F}}(\gamma) d\mu_F(y)\\
        & = \int_F \widehat{f}^{\widehat{F}}_2(u-y,-y)h(y) d\mu_F(y),
    \end{align*}
where $\widehat{f}^{\widehat{F}}_2$ means that we are taking $\widehat{F}$-Fourier transform on the second variable of the function $f$. Consequently, the integral $\int_G f(x) W(x) d\mu(x)$ is an integral operator acting on $L^2(F)$ with the kernel $K(u,y) = \widehat{f}^{\widehat{F}}_2(u-y,-y)$. When $F = \Real^n$, it is clear that $K \in \mc S(\Real^{2n})$ for $f \in \mc S(G)$. When $F = \tor$ we can see that $K \in C^\infty(\tor^2)$ for $f \in \mc S(G)$ since $\F_{\tor^2}(\widehat{f}^{\widehat{F}}_2)(m,n) = \widehat{f_n}^\tor(m)$, $m,n\in \z$. Thus, we can conclude that the corresponding integral operator is a trace class operator (\cite[p.120-121]{GK} and \cite{BR88}) in both of the cases.
\end{proof}

We close this section with the proofs of Proposition \ref{prop-Wigner-properties} on the properties of Wigner functions and Theorem \ref{thm-Clifford covariance}, the Clifford covariance of Wigner functions.

\begin{proof}[Proof of Proposition \ref{prop-Wigner-properties}]
(1) The fact that $\chi_\rho(0)=1$ is clear from definition of characteristic functions and $\int_G \W_\rho(\yv)d\widehat{\mu}(\yv) = \chi_\rho(0)$ for integrable $\W_\rho$ is from the Fourier inversion.

\vspace{0.5cm}

(2) We note the integral formula $\W_\rho(\yv) = \int_G{\rm Tr}(\rho W(\xv)^*)\overline{\Phi(\yv)(\xv)}d\mu(\xv)$. Then the conclusion follows from the fact that ${\rm Tr}(\rho W(-\xv)^*) = {\rm Tr}(\rho W(\xv)) = \overline{{\rm Tr}(\rho W(\xv)^*)}$. Note that the first equality is from the assumption that $\sigma$ is normalized.
\end{proof}

\begin{proof}[Proof of Theorem \ref{thm-Clifford covariance}] 
Let $\rho$, $\yv$, $S$ and $U$ as in Theorem \ref{thm-Clifford covariance}. Then, for $\xv \in G$ we have
    \begin{align*}
        \chi_{U\rho U^*}(\xv) & = \text{\rm Tr}(U\rho U^*W(\xv)^*) = \text{\rm Tr}(\rho U^*_SW(\yv)^*W(\xv)^*W(\yv)U_S)\\
        & = \Phi(\yv)(\xv)\text{\rm Tr}(\rho U^*_S W(\xv)^*U_S) =\Phi(\yv)(\xv)\text{\rm Tr}(\rho W(S^{-1}\xv)^*)\\
        &=\Phi(\yv)(\xv)\chi_\rho(S^{-1}\xv).
    \end{align*}
Thus, we have for $\zv\in G$ that
    \begin{align*}
        \W_{U\rho U^*}(\zv) & = \int_G \chi_\rho(S^{-1}\xv)\Phi(\yv)(\xv)\overline{\Phi(\zv)(\xv)}d\mu(\xv)\\
        & = \int_G \chi_\rho(S^{-1}\xv)\overline{\Phi(\zv-\yv)(\xv)}d\mu(\xv)\\
        & = C^{-1}_S \int_G \chi_\rho(\xv)\overline{\Phi(\zv-\yv)(S\xv)}d\mu(\xv)\\
        & = C^{-1}_S \int_G \chi_\rho(\xv)\overline{\Phi(S^{-1}(\zv-\yv))(\xv)}d\mu(\xv)\\
        & = C^{-1}_S \W_\rho(S^{-1}(\zv-\yv)).
    \end{align*}
Here, we use the fact that $\mu \circ S$ is another Haar measure on $G$, which guarantees the existence of the constant $C_S>0$.   
\end{proof}

\section{Symmetry groups $Sp(G,\sigma)$, $\mc C(G,\sigma)$ and $\mathcal{C}_{\rm gen}(G,\sigma)$} \label{subsec-symmetry-groups}

First, let us elaborate that $\mc C(G,\sigma)$ and $\mathcal{C}_{\rm gen}(G,\sigma)$ actually define groups.

\begin{prop}
Let $U_i$, $i=1,2$ be (resp. generalized) Clifford operations on $(G,\sigma)$. Then, $U_1U_2$ is another (resp. generalized) Clifford operation on $(G,\sigma)$ upto a phase factor.
\end{prop}
\begin{proof}
We only check the case of Clifford operations since the other case can be obtained similarly. Suppose that $U_i$, $i=1,2$ are of the form $U_i = W(\yv_i)U_{S_i}$ for $\yv_i \in G$, $S_i \in Sp(G,\sigma)$, $i=1,2$. Then, we have for any $\xv\in G$ that
    \begin{align*}
        U_1U_2W(\xv)U^*_2U^*_1
        & = W(\yv_1)U_{S_1}\left(W(\yv_2)U_{S_2}W(\xv)U^*_{S_2}W(\yv_2)^*\right)U^*_{S_1}W(\yv_1)^*\\
        & = \Phi(\yv_2)(S_2\xv)W(\yv_1)U_{S_1}W(S_2\xv)U^*_{S_1}W(\yv_1)^*\\
        & = \Phi(\yv_1)(S_1S_2\xv)\Phi(\yv_2)(S_2\xv)W(S_1S_2\xv)\\
        & = \Phi\left((S_1S_2)^{-1}(\yv_1+S_1\yv_2)\right)(\xv)W(S_1S_2\xv).
    \end{align*}
Now we compare the above with \eqref{eq-Clifford-character} and appeal to the Schur's lemma for twisted representations to conclude that $U_1U_2$ and  $U_{S_1S_2}W(\yv_1 + S_1\yv_2)$ coincide upto a phase factor. Note that the latter is a Clifford operations on $(G,\sigma)$.
\end{proof}

\begin{rem}\label{rem-abs-Clifford-semidirect}
From the above proof we can read out that the underlying group law is $(\yv_1, S_1)\cdot (\yv_2, S_2) = (\yv_1 + S_1\yv_2, S_1S_2)$ for $S_i \in Sp(G,\sigma)$ and $\yv_i \in G$, $i=1,2$. Thus, $\mc C(G,\sigma) \cong G \rtimes Sp(G,\sigma)$, and
we actually get the conlusion of Proposition \ref{prop-abs-Clifford-semidirect}.
\end{rem}

Let us continue with a simple observation on the Weyl operators.

\begin{lem} \label{lem-Weyl-independence}
    The set $\{W(\xv):\, \xv\in G\}$ is linearly independent.
\end{lem}
\begin{proof}
    Suppose not. Then we can choose the minimal $n\geq 2$ such that there exist distinct $\xv_j\in G$ and nonzero constants $a_j$, $j=1,\cdots, n$ satisfying
    \begin{equation} \label{eq-linear1}
        a_1W(\xv_1)+\cdots+a_nW(\xv_n)=0.
    \end{equation}
    Taking conjugation with repsect to $W(\yv)$ we get
        $$\overline{\Phi(\xv_1)(\yv)}a_1W(\xv_1)+\overline{\Phi(\xv_2)(\yv)}a_2W(\xv_2)+\cdots+\overline{\Phi(\xv_n)(\yv)}a_nW(\xv_n)=0$$
    and consequently
    \begin{equation} \label{eq-linear2}
        a_1W(\xv_1)+\overline{\Phi(\xv_2-\xv_1)(\yv)}a_2W(\xv_2)+\cdots+\overline{\Phi(\xv_n-\xv_1)(\yv)}a_nW(\xv_n)=0
    \end{equation}
    by multiplying $\Phi(\xv_1)(\yv)$ in both sides. Combining \eqref{eq-linear1} and \eqref{eq-linear2} we get
    $$(1-\overline{\Phi(\xv_2-\xv_1)(\yv)})a_2W(\xv_2)+\cdots+(1-\overline{\Phi(\xv_n-\xv_1)(\yv)})a_nW(\xv_n)=0$$
    with the left hand side $\leq n-1$-terms. This contradicts the minimality of $n$ since we can take an appropriate $\yv\in G$ such that $1-\overline{\Phi(\xv_2-\xv_1)(\yv)}\neq 0$, which is thanks to the fact that $\sigma$ is a Heisenberg multiplier.
\end{proof}
Note that lemma \ref{lem-Weyl-independence} does not necessarily hold if $\sigma$ is not a Heisenberg multiplier. If $\sigma\equiv 1$, for example, every irreducible $\sigma$-representation is actually a character $\eta\in \widehat{G}$. However, the set $\{\eta(\xv):\, \xv\in G\}\subseteq \Comp$ is never linearly independent.

\color{black}

Now we have a better understanding on the elements of the group $\mathcal{C}_{\rm gen}(G,\sigma)$.

\begin{prop} \label{prop-Clifford-injectivity}
    Let $U\in \mathcal{U}(\Hi)$ be a generalized Clifford operation on $(G, \sigma)$ with the associated maps $\xi:G\to \tor$ and $S:G\to G$ satisfying \eqref{eq-Clifford}. Then, $S$ is an injective homomorphism satisfying
    \begin{equation}\label{eq-S-relation}
        \xi(\xv)\xi(\yv)\sigma(S\xv,S\yv)=\xi(\xv+\yv)\sigma(\xv,\yv)
    \end{equation}
    and
    \begin{equation}\label{eq-S-relation2}
        \sigma(S\xv, S\yv)\overline{\sigma(S\yv,S\xv)}=\sigma(\xv,\yv)\overline{\sigma(\yv,\xv)}.
    \end{equation}
    In particular, $S$ is an isomorphism if $G$ is a finite group.
\end{prop}

\begin{proof}
    Let $\xv,\yv\in G$. Comparing the terms $UW(\xv)W(\yv)U^*=UW(\xv)U^*UW(\yv)U^*$ and $UW(\xv+\yv)U^*$, we have
        $$\xi(\xv)\xi(\yv)\sigma(S(\xv),S(\yv))W(S(\xv)+S(\yv))=\xi(\xv+\yv)\sigma(\xv,\yv)W(S(\xv+\yv)).$$
    By lemma \ref{lem-Weyl-independence}, we have $S(\xv)+S(\yv)=S(\xv+\yv)$, which means that $S$ is a homomorphism, and \eqref{eq-S-relation}. By swapping the variables $\xv$ and $\yv$ and taking quotient, we get \eqref{eq-S-relation2}. For the injectivity of $S$ we consider $\xv\in G$ with $S\xv=0$. Then, we have $UW(\xv)U^*=\xi(\xv)I=\xi(\xv)UW(0)U^*$ and consequently $\xv=0$ by lemma \ref{lem-Weyl-independence}. Finally, if $G$ is a finite group, injectivity of $S$ also implies surjectivity.
\end{proof}

Now we can prove Proposition \ref{prop-abs-Clifford} using the above.

\begin{proof}[proof of Proposition \ref{prop-abs-Clifford}]
    $(2)\Leftrightarrow(3)$ directly follows from  \eqref{eq-S-relation}, and $(1)\Rightarrow(2)$ is trivial. Thus it suffices to show $(2)+(3)\Rightarrow(1)$. Since $\sigma$ is a Heisenberg multiplier, we can choose unique $\yv_0\in G$ such that  $\xi(\xv)=\Phi(\yv_0)(\xv)=\sigma(\yv_0,\xv)\overline{\sigma(\xv,\yv_0)}$ for $\xv \in G$. Setting $\yv=S\yv_0$, we have $W(\yv)U_SW(\xv)U_S^*W(\yv)^*=\Phi(\yv)(S\xv)W(S\xv)=\xi(\xv)W(S\xv)$, and therefore $U=W(\yv)U_S$ up to a phase. Hence $U$ is a Clifford operation.
\end{proof}

We also have the partial converse of Proposition \ref{prop-Clifford-injectivity}.

\begin{prop} \label{prop-Clifford-existence}
    If an automorphism $S:G\to G$ satisfies \eqref{eq-S-relation2}, then there is a generalized Clifford operation $U$ on $(G, \sigma)$ and a Borel map $\xi$ related by \eqref{eq-Clifford}.
\end{prop}

\begin{proof}
    We let $m(\xv,\yv):=\frac{\sigma(\xv,\yv)}{\sigma(S\xv,S\yv)}$, which is a 2-cocycle satisfying $m(\xv,\yv)=m(\yv,\xv)$ for $\xv,\yv \in G$. By \cite[Lemma 3]{DigernesVaradarajan04} we know that the 2-cocycle $m$ is trivial, i.e. there exists a Borel function $\xi:G\to \tor$ such that $m(\xv,\yv)=\xi(\xv)\xi(\yv)/\xi(\xv+\yv)$. Consequently, a map $\xv\in G \mapsto \xi(\xv)W(S\xv)$ is a $\sigma$-representation, which is also irreducible since $S$ is an isomorphism. Now we appeal to Stone-von Neumann-Mackey theorem (\cite [Theorem 4]{DigernesVaradarajan04}) to get the generalized Clifford operation we want.
\end{proof}

Before proceeding any further, we verify what we have missed for the 2-cocycles for the Fermionic system and the mixed spin systems.

\begin{lem} \label{lem-fermion-Heisenberg}
    For any choice of $\eps:\{1,\ldots,n\}\times\{1,\ldots,n\}\to \{\pm 1\}$ as in Section \ref{subsec-mixed-spin}, $\sigma_{\eps}$ is a Heisenberg multiplier. In particular, $\sigma_{\rm fer}=\sigma_{-1}$ is a Heisenberg multiplier.
\end{lem}

\begin{proof}
We have the formula $\Phi_{\eps}(\xv)(\yv):=\sigma_{\eps}(\xv,\yv)\overline{\sigma_{\eps}(\yv,\xv)}=(-1)^{\xv^T(\Delta_{\eps}+\Delta_{\eps}^T)\yv}$ for $\xv, \yv\in \z_2^{2n}$. To prove $\Phi_{\eps}:\z_2^{2n}\to \widehat{\z_2^{2n}}\cong \z_2^{2n} $ is an isomorphism, it suffices to show that the matrix $\Delta_{\eps}+\Delta_{\eps}^T\in M_{2n}(\z_2)$ is invertible. For simplicity, we use some temporary notations $E={\Tiny \begin{bmatrix} 1&1\\ 1&1 \end{bmatrix}}$, $\Om={\Tiny \begin{bmatrix} 0&1\\ 1&0 \end{bmatrix}}$ and $\tilde{\eps}_{ij}=\tilde{\eps}(i,j)$. Then $\Delta_{\eps}+\Delta_{\eps}^T$ can be written as
    $$\Delta_{\eps}+\Delta_{\eps}^T=\begin{bmatrix} \Om &\tilde{\eps}_{12}E &\tilde{\eps}_{13}E &\cdots &\tilde{\eps}_{1n}E \\ \tilde{\eps}_{12}E &\Om &\tilde{\eps}_{23}E &\cdots &\tilde{\eps}_{2n}E \\ \tilde{\eps}_{13}E &\tilde{\eps}_{23}E &\Om &\cdots & \tilde{\eps}_{3n}E \\ \vdots &\vdots &\vdots & \ddots &\vdots \\ \tilde{\eps}_{1n}E &\tilde{\eps}_{2n}E &\tilde{\eps}_{3n}E & \cdots &\Om \end{bmatrix}.$$
Now the invertibility of $\Delta_{\eps}+\Delta_{\eps}^T$ follows once we note the matrix identity $P^T(\Delta_{\eps}+\Delta_{\eps}^T)P=\bigoplus^n_{j=1}\Om$, where
    $$P=\begin{bmatrix} I_2 &\tilde{\eps}_{12}E &\tilde{\eps}_{13}E &\cdots &\tilde{\eps}_{1n}E \\ 0 &I_2 &\tilde{\eps}_{23}E &\cdots &\tilde{\eps}_{2n}E \\ 0 &0 &I_2 &\cdots & \tilde{\eps}_{3n}E \\ \vdots &\vdots &\vdots & \ddots &\vdots \\ 0 &0 &0 & \cdots &I_2 \end{bmatrix},$$
from the relations $\Om E=E\Om=E$ and  $E^2=2E=0$.
\end{proof}

\begin{lem}\label{lem-irreducibility-W-eps}
For any choice of $\eps:\{1,\ldots,n\}\times\{1,\ldots,n\}\to \{\pm 1\}$ as in Section \ref{subsec-mixed-spin} the map $W_\eps$ from \eqref{eq-W-eps} is actually the unique (upto unitary equivalence) irreducible unitary $\sigma_\eps$-representation of the group $\z^{2n}_2$.
\end{lem}
\begin{proof}
We only need to check the irreducibility of $W_\eps$, which we appeal to the twisted version of Schur's lemma. Let $X$ be an intertwiner of $W_\eps$, then we know that $X$ commutes with each of $W_\eps(\xv)$, $\xv \in \z^{2n}_2$ and consequently with span$\{W_\eps(\xv): \xv \in \z^{2n}_2\} = M_{2^n}(\Comp)$, so that it is a scalar multiple of the identity. The latter equality comes from the fact that $\{W_\eps(\xv): \xv \in \z^{2n}_2\}$ forms an orthogonal basis of $M_{2^n}(\Comp)$ with respect to the trace inner product.
\end{proof}

Now we determine the symmetry groups (i.e. symplectic and generalized Clifford groups) in Table ~1 and ~2.  

\begin{prop}
\label{prop-symplectic-group}
We have $Sp(\Real^{2n}, \sigma_{\rm can})=\{S\in M_{2n}(\Real):S^T LS=L\}$ where $L={\Tiny \begin{bmatrix} 0&0\\I_n&0\end{bmatrix}}$, and $Sp(\Real^{2n}, \tilde{\sigma}_{\rm can})=Sp_{2n}(\Real)\subseteq SL_{2n}(\Real)$
\end{prop}
\begin{proof}
Note that $\sigma_{\rm can}(\xv)=\exp(i\xv^T L \yv)$, $\xv, \yv \in \Real^{2n}$. If $S\in Sp(\Real^{2n}, \sigma_{\rm can})$, then the relation $\exp(i\xv^TS^T LS\yv)=\exp(i\xv^T L\yv)$ for all $\xv, \yv\in \Real^{2n}$ implies
    $$\xv^T(S^T LS)\yv=\xv^T L \yv+2\pi n(\xv,\yv)$$
for some function $n:\Real^{2n}\times\Real^{2n}\to\z$ with $n(0,0) = 0$. Since $S$ is continuous, so is $n$. Thus $n\equiv 0$ and $S^TL S=L$.
     
Conversely, suppose $S\in M_{2n}(\Real)$ and $S^TL S=L$. Then, $\sigma_{\rm can}$-preserving property is clear and we only need to check $S\in SL_{2n}(\Real)$. Indeed, we also have $S^TL^T S=L^T$ so $S^TJ S=S^T(L^T-L)S=L^T-L=J$. By taking pfaffian on both sides of $S^TJ S = J$ we have $\det(S)=1$  (Recall that the pfaffian ${\rm pf} (A)$ of $A=(a_{ij})=-A^T\in M_{2n}(\Real)$ is given by
${\rm pf} (A):=\frac{1}{n!2^n}\sum_{\pi\in S_{2n}}a_{\pi(1)\pi(2)}\cdots a_{\pi(2n-1)\pi(2n)},$
where $S_{2n}$ is the symmetric group of degree $2n$, and we have ${\rm pf}(BAB^T)=\det (B){\rm pf}(A)$ for another $B\in M_{2n}(\Real)$. Moreover, we clearly have ${\rm pf} (J)=1$). Now the rest of the statements can be shown similarly.
\end{proof}

\begin{rem}\label{rem-symplectic-discrete}
By similar arguments as the above we have the following:

$\begin{cases}Sp(\z^{2n}_d, \sigma_{\rm can})=\{S\in M_{2n}(\z_d)|S^T L S=L\}\subseteq SL_{2n}(\z_d),\; \text{where $L={\Tiny \begin{bmatrix} 0&0\\I_n&0\end{bmatrix}}$},\\
Sp(\z^{2n}_d, \tilde{\sigma}_{\rm can})=Sp_{2n}(\z_d) \subseteq SL_{2n}(\z_d),\; \text{$d\geq 3$, odd integer},\\
Sp(\z^{2n}_2, \sigma_{\rm fer})=\{S\in M_{2n}(\z_2)|S^T \Delta S=\Delta\},\\
Sp(\z^{2n}_2, \sigma_{\epsilon})=\{S\in M_{2n}(\z_2)|S^T \Delta_{\epsilon} S=\Delta_{\epsilon}\}\subseteq SL_{2n}(\z_2).
\end{cases}$
\end{rem}

For low rank cases we were only able to determine the following symplectic groups via tedious calculations, which we omit the details.

\begin{prop} \label{prop-low-rank-symplectic}
We have
$\begin{cases}Sp(\z^2_2, \sigma_{\rm fer})=\{I\},\\
Sp(\z^2_2, \tilde{\sigma}_{\rm fer}) = \{I, S, S^2\}\cong \z_3,\; \text{where $S={\Tiny \begin{bmatrix}0 & 1\\ 1 & 1\end{bmatrix}}$},\\
Sp(\z^4_2, \sigma_{\rm fer}) = \left\{I, S, S^2\right\}\cong \z_3,\; \text{where $S={\Tiny \begin{bmatrix}1&1&0&0\\ 0&1&1&0\\ 0&1&0&0\\ 0&1&0&1 \end{bmatrix}}$}.
\end{cases}$
\end{prop}

Now we consider the case $G = \tor \times \z$. We first determine its (topological) automorphism group. Here, we identify $\tor \cong [0,1)$ with the $\z$-modular addition, i.e. $\alpha = \beta \;\;(\textrm{mod}\, \z)$ means that $\alpha - \beta 
\in \z$.
\begin{lem}\label{lem-aut-tor-z}
The topological automorphism group $Aut(\tor \times \z)$ is isomorphic to the quotient of the matrix group $\{{\footnotesize\begin{bmatrix}m&\alpha\\0&k\end{bmatrix}}: (\alpha, m, k)\in \Real \times \{\pm 1\}^2 \}$ with respect to the subgroup $\{{\footnotesize\begin{bmatrix}1&\alpha\\0&1\end{bmatrix}}: \alpha\in \z \}$. More precisely, for any $S\in Aut(\tor \times \z)$ there is $(\alpha, m, k)\in \tor \times \{\pm 1\}^2$ such that $S(\theta, n) = (m\theta + n\alpha, nk)$, $(\theta, n)\in \tor \times \z$.
\end{lem}
\begin{proof}
For $S\in Aut(\tor \times \z)$ we set $S(0,1) = (\alpha, k) \in \tor \times \z$. Since $S(\cdot, 0): \tor \to \tor\times \z$ is a continuous homomorphism with $S(0,0) = (0,0)$, so that Ran$S(\cdot, 0) \subseteq \tor \times \{0\}$. Thus, $S(\cdot, 0)$ is a character on $\tor$, so that there is $m\in \z$ such that $S(\theta,0)=(m\theta,0)$, $\theta\in [0,1)$. Then, we have
    $$S(\theta,n) = S(\theta,0) + S(0,n) = (m\theta+n\alpha, nk),\; (\theta,n)\in \tor\times \z.$$
The bijectivity of $S$ imply that $m, k \in \{\pm 1\}$ and the choice of $\alpha$ can be arbitrary. The remaining parts are now straightforward.
\end{proof}

\begin{rem}\label{rem-measure-preserving}
It is straightforward to see that any element in $Aut(\tor \times \z)$ is a $\mu$-preserving map on $\tor \times \z$.
\end{rem}

\begin{thm}\label{thm-symplectic-tor-z}
    We have $Sp(\tor \times \z, \sigma_{\rm can})=\{\pm id_{\tor\times\z}\}$ and $Sp(\tor \times \z, \tilde{\sigma}_{\rm can})=\{id_{\tor\times\z}\}$.
\end{thm}
\begin{proof}
We begin with $S\in Aut(\tor \times \z)$ associated with $(\alpha, m, k)\in \tor \times \{\pm 1\}^2$. If $S\in Sp(\tor\times \z, \sigma_{\rm can})$, then it is straightforward to see that the $\sigma_{\rm can}$-preserving property says that
    $$\theta = mk\theta + nk\alpha \;\;(\textrm{mod}\, \z)$$
for any $\theta \in \tor$ and $n\in \z$. Thus, we have $mk=1$ 
and $\alpha = 0$, where the cases $(m,k)=(1,1)$ and $(m,k)=(-1, -1)$ correspond to $S=id_{\tor\times\z}$ and $S=-id_{\tor\times\z}$, respectively.

Now we move to the second case, namely $\tilde{\sigma}_{\rm can}$-preserving property of $S$. Recall that $\xi$ was chosen to be $\xi(\theta, n)=e^{\pi in\{\theta\}}$ with the fractional part $\{x\}=x-\lfloor x\rfloor$. Note that we have $\xi(\theta, n)^2=e^{2\pi in\theta}$, which is much simpler than the expression for $\xi$ itself. 
Hence the equation $\tilde{\sigma}_{\rm can}(S(\theta, n), S(\theta', n'))^2=\tilde{\sigma}_{\rm can}((\theta, n), (\theta', n'))^2$ implies
    $$n'\theta - n\theta' = mkn'\theta -mkn\theta' \;\;(\textrm{mod}\, \z)$$
for any $\theta, \theta' \in \tor$ and $n, n'\in \z$, which still deduces $mk=1$.

When $m=k=1$ and $\alpha\in (0,1)$, the $\tilde{\sigma}_{\rm can}$-preserving property is equivalent to
    $$\frac{\xi(\theta, n)\xi(\theta',n')}{\xi(\theta+\theta',n+n')}=\frac{\xi(\theta+n\alpha,n)\xi(\theta'+n'\alpha,n')}{\xi(\theta+\theta'+(n+n')\alpha,n+n')}e^{2\pi inn'\alpha}$$
for all $\theta, \theta'\in [0,1)$ and $n,n'\in \z$. Since $0<\alpha<1$, we may select $n=1, n'=0, \theta=0$, and $\theta'=1-\alpha\in (0,1)$ to have
    $$-1=\xi(1-\alpha, 1)\xi(\alpha, 1)=\xi(1,1)=\xi(0,1)=1,$$
a contradiction. When $m=k=-1$ and $\alpha\in [0,1)$, we similarly have
    $$\frac{\xi(\theta, n)\xi(\theta',n')}{\xi(\theta+\theta',n+n')}=\frac{\xi(-\theta+n\alpha,-n)\xi(-\theta'+n'\alpha,-n')}{\xi(-(\theta+\theta')+(n+n')\alpha,-(n+n'))}e^{-2\pi inn'\alpha}.$$
Choosing $n=1, n'=0, \theta=0$ and $\theta'=\frac{1+\alpha}{2}$, we get another contradiction
    $$1=\frac{\xi\left(\frac{1+\alpha}{2},1\right)\xi(\alpha,-1)}{\xi\left(\frac{-1+\alpha}{2},-1\right)}=\frac{\xi\left(\frac{1+\alpha}{2},1\right)\xi(\alpha,-1)}{\xi\left(\frac{1+\alpha}{2},-1\right)}=-1$$
for any $\alpha\in [0,1)$, which means that the case $m=k=1$, $\alpha=0$, i.e. $S=id_{\tor\times\z}$, is the only possibility we have.
\end{proof}

Let us turn our attention to the case of generalized Clifford group $\mathcal{C}_{\rm gen}(G,\sigma)$ and begin with the following easy observation.

\begin{lem} \label{lem-Clifford-similarity}
Let $\sigma_1$ and $\sigma_2$ be two Heisenberg multipliers on G, and let $W_j$, $j=1,2$, be the irreducible $\sigma_j$-representation of G, respectively. Suppose there exists a Borel map $\eta:G\to \tor$ and a homeomorphism $T:G\to G$ such that
$$W_2(\xv):=\eta(\xv)W_1(T\xv), \;\; \xv\in G.$$
Then $\mc C_{\rm gen}(G, \sigma_1)=\mc C_{\rm gen}(G,\sigma_2)$. In particular, if $\sigma_1$ and $\sigma_2$ are similar 2-cocyles on $G$, then $\mc C_{\rm gen}(G,\sigma_1)=\mc C_{\rm gen}(G,\sigma_2)$.
\end{lem}

\begin{proof}
Suppose $U$ is a generalized Clifford operation on $(G, \sigma_2)$ with the associated maps $\xi$ and $S$, i.e. $UW_2(\xv)U^*=\xi(\xv)W_2(S\xv)$, $\xv\in G$. Then we have 
    $$UW_1(\xv)U^*=\left[\xi(T^{-1}\xv)\frac{\eta(ST^{-1}\xv)}{\eta(T^{-1}\xv)}\right]W_1(TST^{-1}\xv),\; \xv\in G,$$ 
showing $\mc C_{\rm gen}(G, \sigma_2) \subset \mc C_{\rm gen}(G,\sigma_1)$. Since $T$ is invertible, we also have the reverse inclusion, so that $\mc C_{\rm gen}(G,\sigma_1)=\mc C_{\rm gen}(G,\sigma_2)$.
    
\end{proof}


\begin{thm} We have
    $\begin{cases}\mc C_{\rm gen}(\Real^{2n}, \sigma_{\rm can})=\mc C_{\rm gen}(\Real^{2n}, \tilde{\sigma}_{\rm can})=\mc C(\Real^{2n}, \tilde{\sigma}_{\rm can}),\\
    \mc C_{\rm gen}(\z_d^{2n},\sigma_{\rm can})=\mc C_{\rm gen}(\z_d^{2n}, \tilde{\sigma}_{\rm can})=\mathcal{C}(\z_d^{2n}, \tilde{\sigma}_{\rm can}) \end{cases}$
for an odd integer $d\geq 3$.    
\end{thm}

\begin{proof}
We first consider the case of $\Real^{2n}$. By Lemma \ref{lem-Clifford-similarity}, it suffices to show $\mc C_{\rm gen}(\Real^{2n}, \tilde{\sigma}_{\rm can})=\mc C(\Real^{2n}, \tilde{\sigma}_{\rm can})$. Consider a generalized Clifford operation $U$ associated with the maps $\xi$ and $S$. Proposition \ref{prop-Clifford-injectivity} tells us that $S$ is a homomorphism, in other words, a $\z$-linear map, which in turn is a $\mathbb{Q}$-linear map. The continuity of $S$ actually means that $S$ is $\Real$-linear, so that we have $S\in M_{2n}(\Real)$. Now we get $S\in Sp_{2n}(\Real)$ from \eqref{eq-S-relation2} and the formula $\tilde{\sigma}_{\rm can}(\xv, \yv)=\exp\left(-\frac{i}{2}\xv^TJ\yv\right)$ as in the proof of Proposition \ref{prop-symplectic-group}. Finally, we appeal to Proposition \ref{prop-abs-Clifford} for the conclusion. The case $\z^{2n}_d$ can be done similarly.
\end{proof}

\begin{thm}
    We have $\mc C(\tor \times \z, \sigma_{\rm can})\subsetneq \mc C_{\rm gen}(\tor \times \z, \sigma_{\rm can})=\mc C_{\rm gen}(\tor \times \z, \tilde{\sigma}_{\rm can})$. More precisely, a continuous map $S:\tor \times \z\to \tor \times \z$ corresponds to some generalized Clifford operation on $(\tor \times \z, \sigma_{\rm can})$ if and only if $S\in Aut(\tor\times\z)$ associated with $(\alpha, m, k)$, as in Lemma \ref{lem-aut-tor-z}, such that $m=k=\pm 1$. 
\end{thm}

\begin{proof}
The equality $\mc C_{\rm gen}(\tor \times \z, \sigma_{\rm can})=\mc C_{\rm gen}(\tor \times \z, \tilde{\sigma}_{\rm can})$ directly follows from Lemma \ref{lem-Clifford-similarity}. For the second statement we begin with a generalized Clifford operation $U$ on ($\tor\times\z, \sigma_{\rm can}$) associated with the maps $\xi$ and $S$. Since $S$ is a homomorphism (Proposition \ref{prop-Clifford-injectivity}), there are $m,k \in \z$ and $\alpha \in [0,1)$ such that $S(\theta,n) = (m\theta+n\alpha, nk),\; (\theta,n)\in \tor\times \z$ as in the proof of Lemma \ref{lem-aut-tor-z}. Now the condition \eqref{eq-S-relation2} tells us that
    $$n\theta' - n'\theta = nk(m\theta' + n'\alpha) - n'k(m\theta + n\alpha))=mk(n\theta'-n'\theta)$$
for any $(\theta,n), (\theta', n')\in \tor\times\z$. Therefore $m=k=\pm 1$ and $S\in Aut(\tor\times\z)$. For the converse, we may appeal to Proposition \ref{prop-Clifford-existence}. This also explains $\mc C(\tor \times \z, \sigma_{\rm can})\subsetneq \mc C_{\rm gen}(\tor \times \z, \sigma_{\rm can})$.
\end{proof}

Recall $P_n=\{A_1\otimes\cdots\otimes A_n\, |\, A_j\in\{I, X, Y, Z\}\}$  and the Clifford group $\mc C_n=\{U\in U(2^n): UP_nU^*\subset \pm P_n\}/\tor$ on $n$-qubit system.
\color{black}

\begin{thm}
\label{prop-Clifford-fermion}
For any choice signs $\eps$ as in Section \ref{subsec-mixed-spin} we have $\mc C_{\rm gen} (\z_2^{2n},\sigma_\eps)=\mc C_{\rm gen} (\z_2^{2n},\tilde{\sigma}_\eps)=\mc C_n$. In particular, we have $\mc C_{\rm gen} (\z_2^{2n},\sigma_{\rm fer})=\mc C_{\rm gen} (\z_2^{2n},\tilde{\sigma}_{\rm fer})=\mc C_n.$ 
\end{thm}

\begin{proof}
Recall the irreducible representations $W_\eps$ and $W_{\rm fer}$. Note that the set $\{W_\eps(\xv):\xv \in \z_2^{2n}\}$ coincides with $P_n$ (regardless of the choice of $\eps$) upto phase factors at each point on $\z_2^{2n}$. This means that $\mc C_n\subseteq \mc C_{\rm gen} (\z_2^{2n},\sigma_\eps)$ and there are $\eta:\z_2^{2n}\to\tor$ and a bijection $T:\z_2^{2n}\to \z_2^{2n}$ such that $W_\eps(\xv)=\eta(\xv)W_{\rm fer}(T\xv), \;\; \xv\in \z_2^{2n}$. By Lemma \ref{lem-Clifford-similarity}, we have
    $$\mc C_{\rm gen} (\z_2^{2n},\sigma_\eps)=\mc C_{\rm gen} (\z_2^{2n},\tilde{\sigma}_\eps)=\mc C_{\rm gen} (\z_2^{2n},\sigma_{\rm fer})=\mc C_{\rm gen} (\z_2^{2n},\tilde{\sigma}_{\rm fer}).$$
Conversely, we consider a generalized Clifford operation $U$ on $(\z_2^{2n},\sigma_{\rm fer})$ associated with $\xi$ and $S$. Let us write $W_{\rm fer}(\xv)=k(\xv)\Sigma(\xv)$ for some $k:\z_2^{2n}\to \tor$ and $\Sigma:\z_2^{2n}\to P_n$. Then, we have
    $$U\Sigma(\xv)U^*=\left[\xi(\xv)\frac{k(S\xv)}{k(\xv)}\right]\Sigma(S\xv),\;\; \xv\in \z_2^{2n}.$$
Since $\Sigma(\xv)^2=I$ for all $\xv\in \z_2^{2n}$ we can see that the coefficient term $\xi(\xv)\frac{k(S\xv)}{k(\xv)}$ must be $\pm 1$ for all $\xv \in \z_2^{2n}$. This proves $UP_nU^*\subset \pm P_n$ and therefore the coset $U\cdot \tor$ belongs to $\mc C_n$.
\end{proof}


\color{black}

\section{Non-negativity of Wigner functions for Fermionic states}

The following easy observation will be quite handy for us to determine the class $\D(\Hi)_{\W\ge0}$, which is due to the inversion formula for the symplectic Fourier transform on $\z^{2n}_2$:
	$\chi(\xv) = 2^{-n}\sum_{\yv \in \z^{2n}_2}(-1)^{\xv^T(\Delta+\Delta^T)\yv}\W(\yv).$

\begin{prop}\label{prop-real-valued}
Let $\rho$ be a $n$-mode fermionic state with $\W_\rho$ being real-valued. Then, the associated characteristic function $\chi=\chi_\rho$ is also a real-valued function. 
\end{prop}

\begin{proof}[{\bf Proof of Theorem \ref{thm-1-mode-fermi-Wigner-unnormal}:}]

(1)  Recall that for $(x_1x_2)\in \z_2^2$,
    $$\chi_{\rho}(x_1x_2)={\rm Tr}(\rho W_{\rm fer}(x_1x_2)^*)={\rm Tr}(\rho \ch_2^{x_2}\ch_1^{x_1})={\rm Tr}(\rho Y^{x_2}X^{x_1}).$$
We can easily see that $\chi = \chi_\rho$ is given by $\chi(00)=a+d=1$, $\chi(11) = i(d-a)$, $\chi(10) = b+c=2{\rm Re}\, b$, and $\chi(01) = i(b-c)=-2{\rm Im}\, b$ and Proposition \ref{prop-real-valued} tells us that $\chi(11) \in \Real$, so that $a = d = \frac{1}{2}$. Now we set $b=x+iy$ with $x,y\in \Real$, then we have
	$$\begin{cases}2\W(00) = \chi(00) + \chi(10) + \chi(01) + \chi(11) = 1 + 2(x-y)\\ 
	2\W(10) = \chi(00) + \chi(10) - \chi(01) - \chi(11) = 1 + 2(x+y)\\
	2\W(01) = \chi(00) - \chi(10) + \chi(01) - \chi(11) = 1 + 2(-x-y)\\
	2\W(11) = \chi(00) - \chi(10) - \chi(01) + \chi(11) = 1 + 2(-x+y).\end{cases}$$
Thus, $\W \ge 0$ if and only if $-\frac{1}{2} \le x\pm y \le \frac{1}{2}$. Note that the condition $-\frac{1}{2} \le x\pm y \le \frac{1}{2}$ implies the condition $|b|^2 = x^2+y^2 \le \frac{1}{4} = ad$. Combining all the observations we made so far, we get (1).
\vspace{0.5cm}

(2) Trivial.
\vspace{0.5cm}

(3) Immediate from the above results.
\end{proof}

\begin{proof}[{\bf Proof of Theorem \ref{thm-2-mode-fermi-Wigner-unnormal}:}]
We will first look into vanishing of characteristic functions at certain points to narrow down on the choice of $\rho$.

Note that $\chi(1100) \in i\Real$ since $\overline{{\rm Tr}(\rho c_2c_1)} = {\rm Tr}([\rho c_2c_1]^*) = {\rm Tr}(c_1c_2\rho) = - {\rm Tr}(\rho c_2c_1)$. From Proposition \ref{prop-real-valued} we get $\chi(1100) = 0$. We similarly get $\chi(\xv) = 0$ for $\|\xv\| = x_1 + x_2 + x_3 + x_4 = 2$ or 3. By the inversion formula, we get
	$$\rho =\frac{1}{4}\sum_{\xv\in\z_2^4}\chi(\xv)W_{\rm fer}(\xv)=\frac{1}{4}(1+a_1\ch_1 + a_2\ch_2 + a_3\ch_3 + a_4\ch_4 + a_5\ch_1\ch_2\ch_3\ch_4)$$
where $a_j \in \Real$ with $|a_j|\le 1$ for $1\le j\le 5$. The restriction on $a_j$'s comes from $\chi(\xv) \in \Real$ for $\|\xv\| = 0, 1$ or 4  and $|\chi(\xv)| = |{\rm Tr}(\rho W(\xv)^*)| \le {\rm Tr}(|\rho|) = 1$. Now we have the associated Wigner function
	$$4\W(\yv) = 1 + (-1)^{\|\yv\|-y_1}a_1 + (-1)^{\|\yv\|-y_2}a_2 + (-1)^{\|\yv\|-y_3}a_3 + (-1)^{\|\yv\|-y_4}a_4 + (-1)^{\|\yv\|}a_5,$$
so that we can readily check that $\W \ge 0$ if and only if the condition \eqref{eq-2-mode-condition} holds via tedious calculations.
	
From the canonical matrix realization of $\ch_j$'s we get \eqref{eq-matrix-rho-2-mode}. Then, checking the determinants of left upper square blocks of $\rho$ (this also demands tedious calculations) allows us to see that $\rho \ge 0$ if and only if $a^2_1 + a^2_2 + a^2_3 + a^2_4 \le 1- a^2_5$, which is implied by the condition \eqref{eq-2-mode-condition}. Combining all the above we get (1).

\vspace{0.5cm}
For (2) we examine the condition ${\rm Tr}(\rho^2) = 1$, which is the same as $a^2_1 + a^2_2 + a^2_3 + a^2_4 + a^2_5 = 3$, which contradicts the condition $a_1^2+a_2^2+a_3^2+a_4^2+a_5^2\leq 1$ from $\rho\ge 0$.

\vspace{0.5cm}
(3) See the proof of Theorem \ref{thm-n-mode-fermi-Wigner-unnormal}.
    
\end{proof}

\begin{proof}[{\bf Proof of Theorem \ref{thm-n-mode-fermi-Wigner-unnormal}:}]
As in the proof of Theorem \ref{thm-2-mode-fermi-Wigner-unnormal} the condition $\W_\rho\ge 0$ tells us that $\chi(\xv) = 0$ for $\|\xv\| = 2$, which means all second moments of $\rho$ vanish. Consequently, the Wick formula \cite[p.4]{bravyi2004lagrangian} says that all the moments of $\rho$ vanish except the zero moment, namely ${\rm Tr}(\rho) = 1 =\chi(\mathbf{0})$, which means that the state is the maximally mixed state by the inversion formula. Conversely, for $2^{-n}I_{2^n}$ the associated Wigner function is constant $2^{-n}$ function, which is clearly non-negative.
\end{proof}

\begin{proof}[{\bf Proof of Theorem \ref{thm-1-mode-fermi-Wigner-normal}:}]
We begin with $\rho = {\footnotesize \begin{bmatrix}a & b \\ c & d\end{bmatrix}}$, then $\chi$ and $\tilde{\chi}$ coincide except the following point: $\tilde{\chi}(11) = d-a$. Then, for $b=x+iy$ with $x,y\in \Real$ we have
    $$\begin{cases}2\widetilde{\W}(00) = 1 + d-a + 2(x-y)\\ 
	2\widetilde{\W}(10) = 1 - (d-a) + 2(x+y)\\
	2\widetilde{\W}(01) = 1 - (d-a) + 2(-x-y)\\
	2\widetilde{\W}(11) = 1 + d-a + 2(-x+y),\end{cases}$$
which leads us to the condition \eqref{eq-1-mode-normalized}.
\vspace{0.5cm}

(2) Trivial.
\vspace{0.5cm}

(3) This is immediate from the fact that a 1-mode fermionic gaussian state is a diagonal matrix and the above result (1).
\end{proof}

\begin{proof}[{\bf Proof of Theorem \ref{thm-2-mode-fermi-Wigner-normal}:}]

(1) We first observe that for $\chi(\xv) \in \Real$ for any $\xv\in \z^{2n}_2$ by Proposition \ref{prop-real-valued}. Since
    $$\overline{\chi(\xv)} = \overline{{\rm Tr}(W(\xv)^*\rho)} = {\rm Tr}(W(\xv)\rho) = \sigma_{\rm fer}(\xv,\xv){\rm Tr}(W(\xv)^*\rho) = \sigma_{\rm fer}(\xv,\xv)\chi(\xv),$$
$\chi(\xv)\ne 0$ implies $\sigma_{\rm fer}(\xv,\xv)=1$ and therefore $\tilde{\chi}(\xv) = \sigma_{\rm fer}(\xv,\xv)\chi(\xv)=\chi(\xv)$. This means that $\chi = \widetilde{\chi}$ in any case, hence $\W=\widetilde{\W}$, which leads us to the conclusion we wanted.

\vspace{0.5cm}

(2) Let $\rho$ be the $n$-mode fermionic gaussian state with the correlation matrix $M$ from \eqref{eq-correlation}. We can readily check that
    \begin{align*}
        2^n \rho
        & = \prod_{j=1}^n (1 + ia_j\ch_{2j-1}\ch_{2j})\\
        & = 1 + i\sum_{1\le j \le n}a_j\ch_{2j-1}\ch_{2j} + i^2\sum_{1\le j_1<j_2 \le n}a_{j_1}a_{j_2}\ch_{2j_1-1}\ch_{2j_1}\ch_{2j_2-1}\ch_{2j_2}\\
        & \;\;\;\;\;\;\; + \cdots + i^n a_1\cdots a_n\ch_1\ch_{2n}.
    \end{align*}
From the above formula we can easily read out the characteristic function $\widetilde{\chi}$. In particular, we can see that $|\widetilde{\chi}(\xv)|$, $\xv \in \z^{2n}_2\backslash \{\mathbf{0}\}$, are of the form $|a_{j_1}\cdots a_{j_k}|$ for $1\le k \le n$, $1\le j_1 < \cdots < j_k \le n$. This, in turn, tells us that $\widetilde{\W}(\yv) \ge \chi(\mathbf{0})-\sum_{\xv\in \z^{2n}_2\backslash \{\mathbf{0}\}}|\chi(\xv)|= 2 - \prod^n_{j=1}(|a_j|+1)$ for any $\yv \in \z^{2n}_2$. This explains the first claim.

For the second claim we take $a_1 = a_2 = 1$ in \eqref{eq-correlation}, then we can readily check that $\widetilde{\W}(0000) = 1 - a_1 - a_2 - a_1a_2 = -2 <0$.
\end{proof}

\section{The case of quantum torus: when the 2-cocycle is not an Heisenberg multiplier}
\label{section:IRA}
In this section we discuss an abstract Weyl-Wigner representation for the quantum torus, which corresponds to the case when the associated 2-cocycle is not an Heisenberg multiplier.

Our phase space is the group $G = \z^2$, the 2-dimensional integer lattice, equipped with the 2-cocycle $\sigma_\theta$ ($\theta \in (0,1)$ irrational) given by
    $$\sigma_\theta((m,n),(m',n')) = e^{-\pi i\theta(mn'-m'n)}, \;\; m,n,m',n'\in \z,$$
which is clearly normalized. We can see that $\sigma_\theta$ s nothing but the restriction of $\sigma_{\rm boson}$ to the subgroup $\sqrt{2\pi\theta}\z^2 \cong \z^2$ of $\Real^2$.

Since $\z^2$ is clearly not a self-dual group, which means that $\sigma_\theta$ is a an Heisenberg multiplier. Consequently, we can not expect that $\sigma_\theta$-representation theory could be simple as in the case of Heisenberg multipliers. Indeed, the twisted group von Neumann algebra ${\rm vN}(\z^2, \sigma_\theta)$ actually becomes the {\em hyperfinite type $II_1$ factor}. In particular, it is not of the form $\mc B (\Hi)$ for some Hilbert space $\Hi$ as in Theorem \ref{thm-twistedPlancherel}. Note that the {\em reduced twisted group $C^*$-algebra $C^*_r(\z^2,\sigma_\theta)$} generated by twisted convolution operators $L^\sigma_f \in \mc B(\ell^2(\z^2))$, $f\in \ell^1(\z^2)$ is nothing but the {\em irrational rotation algebra} or $C(\tor^2_\theta)$, the {\em $C^*$-algebra of the ``continuous functions on the quantum torus"}. The last symbol suggests us the notation $L^\infty(\tor^2_\theta)$ instead of ${\rm vN}(\z^2,\sigma_\theta)$ for consistency.

We can still apply the abstract Weyl-Wigner formalism for this system as follows. We begin with the quantum state $\rho$, which is an element of $L^1(\tor^2_\theta)_+$ (the positive cone of the predual of $L^\infty(\tor^2_\theta)$) satisfying $\tau(\rho)=1$, where $\tau$ is the canonical (normal) trace on $L^\infty(\tor^2_\theta)$. Then, we define the associated characteristic function $\chi_\rho$ on $\z^2$ by
    $$\chi_\rho(m,n) = \tau(\rho \lambda_{\sigma_\theta}(m,n)^*),\;\; m,n\in \z,$$
where $\lambda_{\sigma_\theta}: \z^2 \to L^\infty(\tor^2_\theta)\subseteq \mc B(\ell^2(\z^2))$ is the regular $\sigma_\theta$-representation from \eqref{eq-reg-rep}. In other words, we are replacing the unique $\sigma$-representation $W$ in the Heisenberg multiplier case with $\lambda_{\sigma_\theta}$. See Theorem \ref{thm-twistedPlancherel} for their relationship in the Heisenberg multiplier case.

Now we define the associated Wigner function $\W_\rho$ on $\tor^2$ (the usual 2-torus) as the $\z^2$-Fourier transform of $\chi_\rho$. More precisely, we have
    $$\W_\rho(s,t) := \sum_{m,n\in \z}\chi_\rho(m,n)e^{-2\pi i (ms+nt)},\;\; s,t\in [0,1).$$
We close this section with the following coincidence of quantum and classical expectations:
    $$\tau(\rho A) = \int^1_0\int^1_0 \W_\rho(s,t) \W_A(s,t) dsdt$$
for any quantum state $\rho$ and any quantum observable $A=A^* \in L^\infty(\tor^2_\theta)$. Note that the above identity can be easily checked from the fact that $\tau(\lambda_{\sigma_\theta}(m,n)) = \delta_{(m,n), (0,0)}$ and the Plancherel formula on $\z^2$.   

\color{black}

\end{document}